\documentclass[aps,prx,10pt,twocolumn,nofootinbib,superscriptaddress]{revtex4-2}
\usepackage[T1]{fontenc}
\usepackage{amsmath}
\usepackage{amsthm}
\usepackage{amssymb}
\usepackage[english]{babel}
\usepackage{braket}
\usepackage{graphicx}
\usepackage{color}
\usepackage{dsfont}
\usepackage{comment}
\usepackage{mathbbol}
\usepackage{enumitem}
\usepackage{placeins}
\usepackage{mathrsfs}

\usepackage[colorlinks=true,urlcolor=blue,citecolor=blue,linkcolor=blue]{hyperref}

\newcommand{\Nega}{\mathcal{N}}

\DeclareMathOperator{\Tr}{Tr}
\newcommand{\diff}{\mathop{}\!\mathrm{d}}

\newcommand{\A}{\mathcal{A}}

\newtheorem{theorem}{Theorem}

\newtheorem{lemma}{Lemma}

\newtheorem{definition}{Definition}
\newtheorem{proposition}{Proposition}

\newcommand{\Uliege}{\affiliation{Institut de Physique Nucléaire, Atomique et de Spectroscopie, CESAM, University of Liège
\\ B-4000 Liège, Belgium}}

\begin{document}

\title{Total, quantum, and classical measures of anticoherence for mixed spin states}
\author{J\'er\^ome Denis}
\email{jdenis@szu.edu.cn}
\affiliation{Institute of Quantum Precision Measurement, State Key Laboratory of Radio Frequency Heterogeneous Integration, College of Physics and Optoelectronic Engineering, Shenzhen University, Shenzhen 518060, China}
\Uliege
\author{Tara Lacaille}
\Uliege

\author{John Martin}
\email{jmartin@uliege.be}
\Uliege
\author{Eduardo Serrano-Ens\'astiga}
\email{ed.ensastiga@uliege.be}
\Uliege

\date{May 27, 2026}

\begin{abstract}
Anticoherent spin states have isotropic low-order spin moments and are relevant to direction-independent metrology and quantum reference-frame alignment. In contrast to pure states, for mixed states such isotropy may originate either from genuine quantum correlations or from classical statistical mixing. We introduce an axiomatic framework for mixed-state $t$-anticoherence based on the symmetric qubit embedding. We distinguish \emph{total} $t$-anticoherence, non-decreasing under SU(2)-covariant channels, from \emph{quantum} $t$-anticoherence, defined as a resource monotone relative to a chosen total measure and constrained to coincide with it on pure states. This yields a \emph{classical} contribution as their difference. We construct total measures based on reduced-state purity, Hilbert--Schmidt distance, and cumulative multipoles, and we discuss fidelity-based total candidates.
We construct quantum counterparts via convex-roof extensions of pure-state functionals tied to bipartite entanglement in the symmetric sector. We provide explicit mixed-state examples, identify states with maximal quantum anticoherence supported on anticoherent subspaces, study robustness under particle loss for different types of states, and characterize the trade-off between purity and the maximal achievable anticoherence order.
\end{abstract}

\maketitle


\section{Introduction}

Spin states that lack any apparent spatial orientation, due to quantum superpositions, play an important role in quantum information, quantum metrology, and the theory of quantum reference frames. Such states are particularly relevant whenever no shared spatial reference frame is available, or when one seeks to suppress or average out directional information. This notion is captured by the concept of \emph{anticoherent} (AC) spin states, originally introduced for pure states as those whose low-order spin moments are isotropic~\cite{Zim:06}. In particular, anticoherence can be viewed as the opposite of spin coherence: whereas coherent spin states are “most classical” in that they point maximally along some direction, AC states suppress directional moments (e.g.\ $\langle\mathbf{J}\rangle=0$ at first order) by suitable quantum superpositions. AC states exhibit remarkable properties, ranging from connections to spherical designs and Majorana constellations to entanglement classification, robust geometric phases and optimal sensitivity bounds in direction-independent metrology~\cite{Rudzinski2024orthonormalbasesof,BjorkKlimovDeLaHozLeuchsSanchezSoto2015,CrannPereiraKribs2010,PhysRevA.92.052333,2020ESE,Baguette2014,Kolenderski_Demkowicz-Dobrzanski_2008,2017Chryssomalakos,Gol.Jam:18,Mar.Wei.Gir:20,2024Ferretti,Rob.Ber:94,Aguilar_2020,Toponomic:22,2025Goldberg,PhysRevResearch.4.013076}.

From an operational perspective, the isotropy of a spin state determines its usefulness as a quantum reference frame. States that break rotational symmetry can serve as tokens for spatial orientation and enable reference-frame alignment~\cite{PhysRevLett.74.1259,Gis.Pop:99,PeresScudoDirection2001,PeresScudoCartesian2001,Peres01072002,CollinsPopescu2004,RevModPhys.91.025001}, while isotropic states are useless for this task. This viewpoint is formalized within the resource theory of $\mathrm{SU}(2)$ asymmetry, where asymmetry with respect to rotations quantifies the ability of a quantum system to encode and transmit directional information \cite{Bartlett2007,Gour2008,Gour2009,ChitambarGour2019}. In this framework, $\mathrm{SU}(2)$-covariant operations represent free operations that cannot increase directional information. A fundamental example is provided by random global rotations, which model the absence of a shared reference frame or deliberate averaging over orientations \cite{Bartlett2007}. AC states, therefore, correspond to highly symmetric states within this resource theory, carrying no orientation information up to a given order.

Beyond their conceptual role, pure-state anticoherence has found concrete applications in several operational settings. In quantum metrology, highly AC states display optimal or near-optimal performance in fidelity-based estimation tasks that do not assume any preferred spatial direction, reflecting their uniform sensitivity to rotations and their close connection to spherical designs~\cite{Mar.Wei.Gir:20,CrannPereiraKribs2010,PhysRevA.111.022435,2025Goldberg}. More generally, AC states have been identified as optimal quantum rotosensors and as enabling multiparameter estimation  of rotational ($\mathrm{SO}(3)$) parameters (e.g.\ Euler angles) in suitable setups~\cite{Kolenderski_Demkowicz-Dobrzanski_2008,2017Chryssomalakos,Gol.Jam:18,Mar.Wei.Gir:20}.     For mixed states, for instance in noisy rotational metrology, mixtures of AC 
states are also suitable for rotational metrological scenarios~\cite{PhysRevA.111.022435}. 

Anticoherence has also been investigated as a marker of \emph{extreme quantumness}: maximally
AC states can form orthonormal bases whose elements exhibit strong multipartite entanglement and pronounced nonclassical behavior~\cite{Rudzinski2024orthonormalbasesof}. Closely related ideas appear in reference-frame alignment protocols: states with strong low-order directional moments can serve as tokens of an absolute direction, whereas anticoherent states suppress such directional bias and are instead well suited to isotropic estimation of rotations or frame misalignment~\cite{Kolenderski_Demkowicz-Dobrzanski_2008}. These results motivate the experimental realization of anticoherent states~\cite{2017Bouchard,2024Ferretti,2026Denis} and similar states such as  random and $t$-design states~\cite{Goldberg2024,2026Bringewatt}.

While the theory of anticoherence for pure states is now well established, a systematic extension to mixed states is still lacking. For mixed states, isotropy of low-order moments may arise either from genuinely quantum correlations or from classical statistical mixing, such as averaging over unknown orientations. As a consequence, a single scalar quantity is generally insufficient to capture the physical content of anticoherence in the mixed-state regime, particularly from the viewpoint of quantum resources and asymmetry. This tension is already visible in existing quantifiers: several anticoherence measures have been proposed to quantify isotropy for arbitrary states~\cite{Bag.Mar:17}, including distance-based constructions. Yet even the maximally mixed state (MMS), arguably the most classical mixed state, appears maximally anticoherent in the sense of low-order isotropy.

In this work, we introduce an axiomatic framework for \emph{mixed-state} $t$-anticoherence that makes this distinction explicit. We identify two complementary notions: (i) \emph{total} $t$-anticoherence, which quantifies the isotropy of order $t$ irrespective of its classical or quantum origin, and (ii) \emph{quantum} $t$-anticoherence, which isolates the intrinsically quantum contribution and behaves as a quantum resource monotone. Operationally, total $t$-anticoherence characterizes the inability of a state to function as a reference frame at resolution $t$, while quantum $t$-anticoherence captures the part of this isotropy that cannot be generated by classical mixing. From a resource-theoretic perspective, this decomposition could also be interpreted as separating the classical and quantum resources required to prepare such states.

These two notions are naturally associated with distinct operational frameworks. Total $t$-anticoherence is closely related to the resource theory of $\mathrm{SU}(2)$ asymmetry and to reference-frame alignment tasks, where random rotations represent free noise. Quantum $t$-anticoherence, on the other hand, is tied to entanglement properties in the symmetric-qubit embedding~\cite{1932Majorana,Aul.Mar.Mur:10,Aul:12,2020ESE,Baguette2014} and quantifies
isotropy arising from genuinely quantum correlations. This separation allows one to define a meaningful classical contribution to isotropy as the difference between total and quantum $t$-anticoherence.

We formulate axioms for both types of measures, construct explicit examples based on distances to the MMS state and on convex-roof extensions, and show that the total contribution always upper bounds the quantum contribution. Finally, we clarify the relation between our framework and existing tensor-based descriptions of angular structure, in particular the cumulative multipole measures defined and studied in
\cite{DeLaHozKlimovBjorkLeuchsSanchezSoto2013,SanchezSotoKlimovDeLaHozLeuchs2013,BjorkKlimovDeLaHozLeuchsSanchezSoto2015}.

The paper is organized as follows. Section~\ref{Sec:axioms} introduces mixed-state $t$-anticoherence and the axiomatic framework for total, quantum, and classical measures. Section~\ref{Sec:constructions} presents explicit constructions of these measures and their properties, including distance-based and cumulative-multipole formulations. Section~\ref{Sec:evaluation} provides explicit evaluations, and Section~\ref{Sec:Applications} discusses applications to representative mixed states, robustness under particle loss, and the trade-off between purity and the maximal achievable anticoherence order. We conclude in Sec.~\ref{Sec.Conc} with some final comments.

\section{Axiomatic framework for mixed-state anticoherence measures}
\label{Sec:axioms}

\subsection{Mixed-state \texorpdfstring{$t$}{t}-anticoherence}

Let $\mathcal H_j$ denote the Hilbert space of a spin-$j$ system, with standard angular-momentum basis
$\{\ket{j,m}\}_{m=-j}^{j}$ (eigenstates of $J^2$ and $J_z$). Using the standard identification
between $\mathcal H_j$ and the fully symmetric subspace of $N=2j$ qubits, any state
$\rho\in \mathcal D(\mathcal H_j)$ defines a symmetric $N$-qubit state $\rho_S$. For
$t\in\{1,\dots,N\}$ we denote by
\begin{equation}
  \rho_t = \Tr_{N-t}\bigl(\rho_S\bigr)
\end{equation}
the reduced state on $t$ qubits, supported on the symmetric subspace of dimension $t+1$.

It is often convenient to characterize spin states in terms of irreducible tensor
operators $\{T_{LM}\}$~\cite{Var.Mos.Khe:88}\footnote{With the Hilbert--Schmidt normalization
$\mathrm{Tr}(T_{LM}^\dagger T_{L'M'})=\delta_{LL'}\delta_{MM'}$.}. Any spin-$j$ state
$\rho$ can be expanded as
\begin{equation}\label{rhoTLMexpansion}
  \rho=\frac{\mathbb 1}{2j+1}
  +\sum_{L=1}^{2j}\sum_{M=-L}^{L} \rho_{LM}T_{LM}.
\end{equation}
The coefficients $\rho_{LM}$ are the multipole moments of $\rho$ and encode its
angular structure at rank $L$.

\begin{definition}\label{defAC}
A spin state $\rho$, pure or mixed, is \emph{$t$-anticoherent} ($t$-AC) if all its
multipole moments up to rank $t$ vanish, namely
\begin{equation}
  \rho_{LM}=0,
  \qquad 1\le L\le t,\quad -L\le M\le L .
\end{equation}
Equivalently, all spin moments up to order $t$ are isotropic.
\end{definition}

Using the standard identification between $\mathcal H_j$ and the fully symmetric
subspace of $N=2j$ qubits, this definition admits an equivalent reduced-state
formulation. If $\rho_S$ denotes the corresponding symmetric $N$-qubit state and
\begin{equation}
  \rho_t=\Tr_{N-t}(\rho_S)
\end{equation}
is its reduced state on $t$ qubits, then
\begin{equation}
  \rho \text{ is } t\text{-AC}
  \quad\Leftrightarrow\quad
  \rho_t=\frac{\mathbb 1_{t+1}}{t+1}.
\end{equation}
Indeed, the partial trace map in the irreducible tensor basis reads~\cite{Den.Mar:22}
\begin{equation}\label{rhotTLM}
\rho_t=\sum_{L=0}^{t}\sum_{M=-L}^{L}
\tfrac{t!}{N!}\sqrt{\tfrac{(N-L)!(N+L+1)!}{(t-L)!(t+L+1)!}}\,
\rho_{LM}\,T_{LM}^{(t)},
\end{equation}
where $T_{LM}^{(t)}$ denotes the irreducible tensor operator acting on the spin-$t/2$ Hilbert space. Thus $\rho_t$ depends only on the coefficients $\rho_{LM}$ with $L\le t$, while all higher-rank components vanish under reduction. Consequently, $\rho_t$ is maximally mixed if and only if all multipole moments of $\rho$ up to rank $t$ vanish.

Let us note two immediate consequences of Definition~\ref{defAC}. First, $t$-anticoherence is monotone in the order: if $\rho$ is $t$-AC, then it is $t'$-AC for every $t'\le t$. This follows either from the nesting of partial traces, $\rho_{t'}=\Tr_{t-t'}(\rho_t)$, together with $\rho_t=\mathbb 1_{t+1}/(t+1)$, or equivalently from~\eqref{rhotTLM}, since vanishing multipoles up to rank $t$ implies vanishing multipoles up to rank $t'$. Second, $t$-anticoherence is stable under particle loss: for any $q> t$, if $\rho_q=\Tr_{N-q}(\rho_S)$, then $(\rho_q)_t=\Tr_{q-t}(\rho_q)=\Tr_{N-t}(\rho_S)=\rho_t$ by the associativity of the partial trace, so $\rho_t=\mathbb 1_{t+1}/(t+1)$ implies $(\rho_q)_t=\mathbb 1_{t+1}/(t+1)$, i.e., $\rho_q$ is again $t$-AC.

Definition~\ref{defAC} reduces to the standard notion of $t$-anticoherence for pure states~\cite{Gir.Bra.Bag.Bas.Mar:15} given by Zimba~\cite{Zim:06}. For mixed states, however, isotropy up to order $t$ can arise either from classical statistical mixing, genuine quantum correlations, or both; this motivates the distinction introduced in this work between total, quantum, and classical measures of anticoherence.

\subsection{Total \texorpdfstring{$t$}{t}-AC measure}
\label{Sec:TAC}

A measure of total $t$-AC quantifies the isotropy, or lack of directional information, of order $t$ independently of its classical or quantum origin\footnote{The definition \eqref{defAC} of $t$-anticoherence is valid for all $t\in\{1,\dots,N\}$. Whenever we consider pure states, we usually assume $t\le N/2$.}.

\begin{definition}\label{def:total_measure}
A function $\A^T_t:\mathcal D(\mathcal H_j)\to[0,1]$ is a \emph{total $t$-AC measure} if it satisfies:
\begin{enumerate}[label=\textit{(T\arabic*)}, ref=T\arabic*]
  \item \label{ax:T1} $\A^T_t(\rho)=1$ iff $\rho$ is $t$-AC;
  \item \label{ax:T2} $\A^T_t(\rho)=0$ iff $\rho$ is pure and coherent;
  \item \label{ax:T3} $\A^T_t(U\rho U^\dagger)=\A^T_t(\rho)$ for all $U\in\mathrm{SU}(2)$;
  \item \label{ax:T4} Monotonicity under $\mathrm{SU}(2)$-covariant noise:
For every $\mathrm{SU}(2)$-covariant channel
$\Phi:\mathcal D(\mathcal H_j)\to\mathcal D(\mathcal H_j)$, i.e.\ every completely
positive trace-preserving map such that
\begin{equation}
  \Phi(U\rho U^\dagger)=U\,\Phi(\rho)\,U^\dagger
\end{equation}
for all $U\in\mathrm{SU}(2)$ and $\rho\in\mathcal D(\mathcal H_j)$, one has
\begin{equation}
\label{Eq.Axiom.4.AT.Phi}
  \A^T_t\!\bigl(\Phi(\rho)\bigr) \ge \A^T_t(\rho)
\end{equation}
for all $\rho\in\mathcal D(\mathcal H_j)$.
\end{enumerate}
\end{definition}

Axioms \eqref{ax:T3}--\eqref{ax:T4} place total $t$-AC in direct correspondence with the resource-theoretic notion of $\mathrm{SU}(2)$-\emph{asymmetry} (or, equivalently, the ability to encode directional information) \cite{Bartlett2007,Gour2008,Gour2009,Marvian2013,Marvian2014}. In this framework, \emph{free operations} are precisely the $\mathrm{SU}(2)$-covariant channels $\Phi$, while \emph{free states} are the $\mathrm{SU}(2)$-invariant states. Since we work on a fixed spin-$j$ irreducible representation, Schur's lemma implies that the unique free state is the MMS state $\rho_0 \equiv \mathbb{1}_{2j+1}/(2j+1)$. Axiom \eqref{ax:T4} expresses that total $t$-anticoherence cannot decrease under free $\mathrm{SU}(2)$ noise that degrades directional information. Operationally, a large value of $\A^T_t(\rho)$ means that $\rho$ is a poor indicator for the alignment of the reference frame at resolution $t$, as its low-order multipole moments are suppressed.

The next lemma presents a basic structural property of $\mathrm{SU}(2)$-covariant channels on a fixed spin sector, which will be used repeatedly in the construction of explicit total measures.

\begin{lemma}[\cite{Holevo2002}]
\label{lem:unital_covariant}
Let $\Phi$ be an $\mathrm{SU}(2)$-covariant channel on $\mathcal H_j$. Then $\Phi$ is unital, i.e.,\ $\Phi(\mathbb 1)=\mathbb 1$, and in particular the MMS state $\rho_0 $ is a fixed point of $\Phi$.
\end{lemma}

\noindent
Lemma~\ref{lem:unital_covariant} explains why the isotropic MMS state will naturally play the role of reference point in the distance- and purity-based construction of AC measures: any $\mathrm{SU}(2)$-covariant channel fixes it.

Random global rotations with a conjugation-invariant probability density,
\begin{equation}\label{eq:random_rotations}
  \rho\ \mapsto\ \int_{\mathrm{SU}(2)} p(U)\,U\rho U^\dagger\,\mathrm dU,
  \qquad p(VUV^\dagger)=p(U),
\end{equation}
form a natural subclass of $\mathrm{SU}(2)$-covariant channels and describe \emph{classical} uncertainty about the laboratory orientation. They model the absence of a shared reference frame or deliberate orientation averaging \cite{Bartlett2007}. However, the complete set of $\mathrm{SU}(2)$-covariant channels is strictly larger: it contains all CPTP maps that commute with the rotation action, including irreversible dissipative processes~\cite{Muller2009,Rivas2013,Den.Mar:22} and measure-and-reprepare (entanglement-breaking) operations, provided they remain \emph{direction insensitive}.

Since the operator space on a fixed spin-$j$ irrep decomposes multiplicity-free into irreducible tensor sectors, covariance implies, by Schur's lemma, that $\Phi$ acts diagonally on irreducible tensor operators,
\begin{equation}\label{fL}
  \Phi(T_{LM}) = f_L\,T_{LM}\qquad \forall\; L,M,
\end{equation}
with $f_L$ real numbers such that $|f_L|\leq 1$; in particular, one always has $f_0=1$ (trace preservation), and complete positivity imposes nontrivial constraints on the admissible values of $(f_1,\dots,f_{2j})$, see~\cite{S0129055X22500210,Aschieri2024}\footnote{For random global rotations of the form \eqref{eq:random_rotations} with $p$ depending only on the rotation angle, the coefficients $f_L$ are the angular averages of the corresponding irreducible characters,
\begin{equation}
f_L=\frac{1}{2L+1}\int_{\mathrm{SU}(2)} p(U)\,\chi^{(L)}(U)\,\mathrm dU,
\end{equation}
where $\chi^{(L)}$ is the character of the spin-$L$ representation. In particular, $f_0=1$ (trace preservation), and for $L\ge 1$ the values $f_L$ quantify how strongly the random rotations damp the rank-$L$ tensor components.}. Therefore, from the expansion \eqref{rhoTLMexpansion}, one has
\begin{equation}
  \Phi(\rho)=\frac{\mathbb 1}{2j+1}
  +\sum_{L=1}^{2j}\sum_{M=-L}^{L} f_L\,\rho_{LM}\,T_{LM}.
\end{equation}
Thus, a covariant channel may attenuate each multipole rank $L$ by an independent factor $f_L$ while treating all components $M$ identically.

\subsection{Quantum \texorpdfstring{$t$}{t}-anticoherence measure}

While total $t$-anticoherence quantifies isotropy independently of its origin, it is useful to single out the contribution stemming from genuinely quantum correlations. For mixed states, this leads to the notion of \emph{quantum} $t$-anticoherence measures, which must behave as resource monotones and remain non-increasing under probabilistic state preparation.

Given a total $t$-anticoherence measure $\A^T_t$ according to Definition~\ref{def:total_measure}, we define a corresponding notion of quantum $t$-anticoherence as follows.
\begin{definition}
\label{def:quantum_measure}
A function $\A^Q_t:\mathcal D(\mathcal H_j)\to[0,1]$ is called a
\emph{quantum $t$-anticoherence measure associated with $\A^T_t$} if it satisfies:
\begin{enumerate}[label=\textit{(Q\arabic*)}, ref=Q\arabic*]

  \item \label{ax:Q1} \emph{Vanishing on separable states:}
  $\A^Q_t(\rho)=0$ for every $\rho$ corresponding to a \emph{fully separable} multiqubit symmetric state $\rho_S$\footnote{
  Equivalently, $\A^Q_t(\rho)=0$ for every symmetric state that is separable across the
  bipartition $t\,|\,N{-}t$, since in the symmetric subspace bipartite separability implies
  full separability. Such symmetric states are of the form
  \begin{equation}
    \rho_S=\int \diff\mu(\ket{n})\, P(n) \bigl(\ket{n}\!\bra{n}\bigr)^{\otimes N}
  \end{equation}
 with $P(n)$ a non-negative function.
  }.

  \item \label{ax:Q2} \emph{Pure-state consistency with the chosen total measure\footnote{
  In particular, if $\A^T_t$ coincides on pure states with a previously introduced pure-state anticoherence measure $\A_t$ in \cite{Bag.Mar:17}, then $\A^Q_t$ coincides with $\A_t$ on pure states as well.}:}
  for all pure states $|\psi\rangle$,
  \begin{equation}
\A^Q_t(|\psi\rangle\langle\psi|)=\A^T_t(|\psi\rangle\langle\psi|).
  \end{equation}

  \item \label{ax:Q3} \emph{$\mathrm{SU}(2)$ invariance: $\forall\;U\in\mathrm{SU}(2)$}
  \begin{equation}
       \A^Q_t(U\rho U^\dagger)=\A^Q_t(\rho).
  \end{equation}

  \item \label{ax:Q4} \emph{Convexity:}
  $\A^Q_t\!\left(\sum_i p_i\rho_i\right)\le \sum_i p_i \A^Q_t(\rho_i).$

  \item \label{ax:Q5} \emph{Average monotonicity under symmetric-sector LOCC:}
For any state $\rho$ and every LOCC protocol across the bipartition $t\,|\,N{-}t$ in the symmetric $N$-qubit embedding, whose normalized outcomes remain supported on the symmetric subspace,
\begin{equation}
  \A^Q_t(\rho)\ge \sum_k p_k\A^Q_t(\rho_k) ,
\end{equation}
where $\rho_k$ is the corresponding spin-state representatives, obtained with probabilities $p_k$.
  \item \label{ax:Q6} \emph{Ordering: $\forall\;\rho\in\mathcal D(\mathcal H_j)$}
  \begin{equation}\label{Eq.TA.greater.QA}
    \A^Q_t(\rho)\le \A^T_t(\rho).
  \end{equation}

\end{enumerate}
\end{definition}

Because, in the symmetric $N$-qubit subspace, separability across any bipartition $t|N-t$ is equivalent to full separability, the natural free states for a “quantum anticoherence” resource theory are exactly the fully separable symmetric states (and their convex mixtures). Conversely, any symmetric state that is not of this form is necessarily genuinely multipartite entangled~\cite{Ichikawa2008ExchangeSymmetry}. For any $q\ge t$, tracing out $N-q$ particles is an LOCC operation. After identifying the retained state $\rho_q$ with a symmetric $q$-qubit state and considering the induced bipartition $t\,|\,q{-}t$, axiom~\eqref{ax:Q5} implies that the corresponding order-$t$ quantum contribution cannot increase under particle loss. By convexity \eqref{ax:Q4} together with $\mathrm{SU}(2)$ invariance \eqref{ax:Q3}, one also finds that the quantum $t$-anticoherence is \emph{non-increasing} under random global rotations \eqref{eq:random_rotations}, in contrast with total $t$-anticoherence, which is \emph{non-decreasing} under the same class of operations by axiom \eqref{ax:T4}.

In the spin-$j$ representation, axiom \eqref{ax:Q5} refers to operations that become LOCC after identifying the spin-$j$ state with a symmetric $N=2j$ qubit state and splitting those constituents into a $t$-particle block and its complement. Physically, these are internal block-local manipulations of the collective spin, assisted by classical feed-forward but unable to generate new quantum correlations across that split; hence they should not increase the genuinely quantum contribution to order-$t$ isotropy.

\subsection{Classical \texorpdfstring{$t$}{t}-anticoherence measure}
The distinction between total and quantum $t$-AC naturally leads to a third, complementary notion that captures the isotropy arising from classical mixing, which we call classical $t$-AC measure.
\begin{definition}
\label{def:classical_measure}
Given a total measure $\A_t^T$ and its associated quantum measure $\A_t^Q$, we define the
\emph{classical $t$-AC measure} $\A_t^C :\mathcal D(\mathcal H_j)\to[0,1]$ by
\begin{equation}
  \A^C_t(\rho) \equiv\A^T_t(\rho)-\A^Q_t(\rho).
\end{equation}
\end{definition}

By axiom \eqref{ax:Q6}, $\A_t^C(\rho)\ge 0$ for all states, and by the pure-state consistency \eqref{ax:Q2}, it vanishes identically on pure states. Hence, $\A_t^C$  quantifies the contribution to isotropy at order $t$ that can be attributed to classical statistical mixing rather than to intrinsically quantum correlations. Axioms \eqref{ax:Q4} and \eqref{ax:Q5} ensure that quantum $t$-anticoherence cannot be increased by classical mixing or by operations that do not generate quantum correlations. In contrast to total $t$-AC, convexity is therefore an essential requirement.

The separation $\A_t^T=\A_t^Q+\A_t^C$ makes explicit that a state may display large (or even maximal) \emph{total} $t$-anticoherence while having no \emph{quantum} contribution, i.e.\ $\A_t^Q=0$ and $\A_t^C=\A_t^T$. The most extreme example is the MMS state $\rho_0$, which is $t$-AC for every $t$ since all multipole moments vanish, yet contains no quantum correlations and is therefore naturally assigned $\A_t^Q(\rho_{0})=0$ (hence $\A_t^C(\rho_{0})=1$) for any possible quantum $t$-AC measure. A second example is provided by the equal incoherent mixture of two opposite spin-coherent states along a fixed axis,
\begin{equation}
\rho=\tfrac12\ket{j,j}\!\bra{j,j}+\tfrac12\ket{j,-j}\!\bra{j,-j},
\end{equation}
which has $\langle \mathbf{J}\rangle=0$ (and thus is perfectly isotropic at first order $t=1$) purely by classical averaging of two antipodal directions. In particular, $\rho$ has $\A_1^C(\rho)=\A_1^T(\rho)=1$ for any 1-AC measure.

\section{Constructions of AC measures}
\label{Sec:constructions}

\subsection{Total AC measures: General constructions}
We begin by constructing several measures of total anticoherence. 

\subsubsection{Purity-based construction}

A particularly simple example of a total $t$-anticoherence measure is obtained from the purity of the reduced symmetric $t$-particle state. We define
\begin{equation}
  \A^{T,\mathcal{P}}_t(\rho)
  \equiv \frac{t+1}{t}\Bigl(1-\operatorname{Tr}(\rho_t^2)\Bigr).
\end{equation}
This expression coincides with the standard purity-based measure of
$t$-anticoherence $\mathcal{A}_t^R$ for pure states~\cite{Bag.Mar:17} and provides a direct extension to mixed states.

\begin{proposition}[Proof in Appendix~\ref{Subsec.Prop1}]
\label{prop1:purity_total}
The quantity $\A^{T,\mathcal{P}}_t(\rho)$ satisfies axioms \eqref{ax:T1}--\eqref{ax:T4} and therefore defines a measure of total $t$-anticoherence.
\end{proposition}
\subsubsection{Distance-based construction}

Distance-based constructions provide another natural class of total $t$-anticoherence candidates. Let $d(\cdot,\cdot)$ be a unitarily invariant operator distance on the state space of the symmetric $t$-qubit sector. We normalize the distance from the MMS state
$\rho_{0}^{(t)} \equiv \mathbb{1}_{t+1}/(t+1)$, by its largest possible value,
\begin{equation}
  K_t=\max_{\sigma\in\mathcal D(\mathcal H_{t/2})} d(\sigma,\rho_0^{(t)}).
\end{equation}
This leads to the normalized distance-based total anticoherence candidate
\begin{equation}
  \A^{T,d}_t(\rho)
  \equiv
  1-\frac{d(\rho_t,\rho_0^{(t)})}{K_t}.
  \label{eq:GeneralDistanceTAC}
\end{equation}
With this convention, states whose $t$-body reduction equals the MMS state have maximal total $t$-anticoherence, while states farthest from the MMS state have minimal value. For the standard distances considered below, whose maximal distance from $\rho_0^{(t)}$ is attained exactly on pure states, this construction satisfies axioms \eqref{ax:T1}, \eqref{ax:T2}, and \eqref{ax:T3}. If the maximum is attained only when $\rho_t$ is pure, then, because $\rho_t$
is a marginal of a symmetric $N$-qubit state, this implies that the original
state is a pure spin-coherent state. Only \eqref{ax:T4} must be proved separately for all $\mathrm{SU}(2)$-covariant channels.

In particular, for a fixed integer $p\geq1$, we have the distances
\begin{equation}
    d_p(\rho,\sigma)=||\rho-\sigma||_p,
    \label{eq:Schatten_Distances}
\end{equation}
where $\rho,\sigma\in\mathcal{D}(\mathcal{H}_j)$, based on the Schatten $p$-norm
\begin{equation}
    ||A||_p=\left(\sum_{i=1}^r\lambda_i^p\right)^{1/p},
\end{equation}
where $\lambda_i$ and $r$ are the singular values and the rank of $A$, respectively.
\begin{proposition}[Proof in Appendix~\ref{Subsec.PropHS_total}]
\label{prop2:HS_total}
The quantity $\A^{T,d}_t(\rho)$ for the distance based on the Schatten $2$-norm (associated with the Hilbert-Schmidt distance) satisfies axioms \eqref{ax:T1}-\eqref{ax:T4} and therefore defines a measure of total $t$-anticoherence.
\end{proposition}
We also find numerical evidence that axiom~\eqref{ax:T4} remains valid for the trace-distance case $p=1$; a proof is left open.

\subsubsection{Fidelity-based construction}
\label{subsubsec:fidelity_based_total}
Another construction is obtained from the Uhlmann-Josza fidelity between the reduced state
$\rho_t$ and the MMS state
$\rho_0^{(t)}=\mathbb{1}_{t+1}/(t+1)$ on the symmetric $t$-qubit subspace:
\begin{equation}
    F(\rho_t,\rho_0^{(t)})
    \equiv
    \left( \Tr\sqrt{\sqrt{\rho_t}\rho_0^{(t)}\sqrt{\rho_t}}\right)^2
    =
    \frac{(\Tr \sqrt{\rho_t})^2}{t+1}.
\label{eq:FidelityReducedStateMMS}
\end{equation}
We choose an affine rescaling 
and define
\begin{equation}\label{Fidelity.ACT}
  \A^{T,F}_t(\rho)
  \equiv \frac{(t+1)F-1}{t}
  =\frac{(\Tr\sqrt{\rho_t})^2-1}{t}.
\end{equation}
Equivalently,
\begin{equation}
\label{BuresACTlin}
\mathcal{A}_t^{T,F}(\rho)
=
\frac{
(t+1)\big[1-\frac12D_{\mathrm B}\big(\rho_t,\rho_0^{(t)}\big)^2\big]^2-1
}{t}
\end{equation}
in terms of the Bures distance
\begin{equation}\label{Bdist}
  D_{\mathrm{B}}(\sigma,\tau)=\sqrt{2\bigl(1-\sqrt{F(\sigma,\tau)}\bigr)}.
\end{equation}
The normalization is such that $\A^{T,F}_t=1$ for $t$-AC states and
$\A^{T,F}_t=0$ for spin-coherent states. Indeed, $F(\rho_t,\rho_0^{(t)})=1$ iff
$\rho_t=\rho_0^{(t)}$, whereas the minimum value $F(\rho_t,\rho_0^{(t)})=1/(t+1)$ is reached
iff $\rho_t$ is pure, which in the symmetric setting implies that $\rho$ is a pure
spin-coherent state. The $\mathrm{SU}(2)$ invariance follows from the unitary
invariance of the fidelity,
$F(U\rho_tU^{\dagger},U\sigma U^{\dagger})=F(\rho_t,\sigma)$, together with
$U\rho_0^{(t)} U^\dagger=\rho_0^{(t)}$. Thus $\A^{T,F}_t$ satisfies axioms~\eqref{ax:T1}--\eqref{ax:T3}.
Concerning axiom~\eqref{ax:T4}, some care is needed. Although a $\mathrm{SU}(2)$-covariant channel $\Phi$ on the full spin-$j$ space induces a linear transformation on the reduced state $\rho_t$, this induced reduced map is not, in general, a completely positive trace-preserving channel acting on the whole state space $\mathcal D(\mathcal H_{t/2})$. Therefore one cannot simply invoke monotonicity of the Uhlmann fidelity for an induced channel fixing $\rho_0^{(t)}$. Numerical evidence suggests that $\A^{T,F}_t$ is nevertheless
non-decreasing under $\mathrm{SU}(2)$-covariant noise, but we leave a proof of axiom~\eqref{ax:T4} for future work.

We conclude this section with a general result, which will prove useful later on the relation between the Uhlmann fidelity and entanglement. Consider a quantum state $|\psi\rangle$ defined on a composite Hilbert space $\mathcal{H}_{AB}$. Then, a measure of entanglement of the state $|\psi\rangle$ on the bipartition $A|B$ is given by the entanglement negativity~\cite{2002Vidal}
\begin{equation}
    \mathcal{N}_A=\sum_{i>j=1}^{d_A}\alpha_i\alpha_j
    \label{eq:Negativity_SchmidtDecomposition}
\end{equation}
where the $\alpha_i$ are the Schmidt coefficients of $|\psi\rangle$ with respect to the split $A|B$. From this definition, we have the following result.

\begin{proposition}[Proof in Appendix~\ref{Subsec.Prop3}]
\label{prop:FidelityNegativityRelation}
    For any composite system $AB$, a pure state $\ket{\psi}\in\mathcal{H}_{AB}$ verifies
    \begin{equation}
        F \left(\rho_A,\rho_0\right)=\frac{1+2\Nega_A(\ket{\psi})}{d_A},
        \label{eq:FidelityReduced_Negativity}
    \end{equation}
    where $\rho_A$ is the reduced state defined on $A$ obtained by tracing out the subsystem $B$ and $\rho_0=\mathbb 1_A/d_A$ is the MMS state in $\mathcal{H}_A$.
\end{proposition}
This result gives a direct correspondence between the fidelity-based AC measures \eqref{Fidelity.ACT} of a pure state and its entanglement with respect to the bipartition $t|N-t$.

\subsection{Quantum AC measures: general constructions via convex-roof}

We now construct quantum $t$-AC measures from pure-state functionals by convex-roof extensions. 

\subsubsection{Purity-based construction}

For pure symmetric states, the quantity
$1-\operatorname{Tr}(\rho_t^2)$ coincides with the linear entropy of entanglement
across the bipartition $t\,|\,N{-}t$.
This motivates defining a quantum $t$-AC measure for mixed states via
the convex-roof construction,
\begin{equation}
  \A^{Q,\mathcal{P}}_t(\rho)
  = \min_{\{p_i,|\psi_i\rangle\}}
     \sum_i p_i\, \A^{\mathcal{P}}_t(|\psi_i\rangle),
\end{equation}
where the minimization runs over all pure-state decompositions of $\rho$.

\begin{proposition}[Proof in Appendix~\ref{Subsec.Prop4}]
\label{prop:purity_quantum_Q1Q6}
The quantity $\A^{Q,\mathcal{P}}_t(\rho)$ satisfies axioms \eqref{ax:Q1}--\eqref{ax:Q6} and therefore defines a quantum measure of $t$-AC associated with the total purity-based measure $\A_t^{T,\mathcal{P}}$.
\end{proposition}

\begin{proposition}[Proof in Appendix~\ref{Subsec.Prop5}]
\label{prop:AtQRRatios}
For all mixed states $\rho\in\mathcal D(\mathcal H_j)$ and all
$t\in\{1,2,\dots,N-1\}$,
\begin{equation}
  \frac{\A^{Q,\mathcal{P}}_t(\rho)}{\A^{Q,\mathcal{P}}_{N-t}(\rho)}
  = \frac{(t+1)(N-t)}{t(N+1-t)} .
\end{equation}
\end{proposition}

\subsubsection{Based on convex distances}

We will use a distance $d$ that is convex in its first argument. That is, for all $\rho_1,\rho_2,\sigma\in\mathcal{D}(\mathcal{H}_j)$ and all $\lambda\in[0,1]$, it satisfies
\begin{equation}
d(\lambda\rho_1+(1-\lambda)\rho_2,\sigma)\leq\lambda d(\rho_1,\sigma)+(1-\lambda)d(\rho_2,\sigma).
\end{equation}
Given such a convex distance $d$, we extend the pure-state quantum $t$-AC measure to mixed states via the convex-roof construction:
\begin{equation}
\A^{Q,d}_t(\rho)
= \min_{\{p_i,|\psi_i\rangle\}}
\sum_i p_i \A^{d}_t(|\psi_i\rangle),
\label{eq:QACMeasure_ConvexDistance}
\end{equation}
where the minimization is taken over all pure-state decompositions of $\rho$, and $\A^{d}_t$ denotes the pure-state $t$-AC measure induced by the distance $d$. When the pure-state functional $\A^d_t(|\psi\rangle)$ is an entanglement monotone across the bipartition $t\,|\,N{-}t$, its convex-roof extension provides a valid quantum $t$-AC measure associated with the corresponding total distance-based quantity.

\begin{proposition}[Proof in Appendix \ref{Subsec.Prop6}]
\label{prop:distance_quantum_Q1Q6}
For any distance $d$ whose associated total quantity $\A_t^{T,d}$ satisfies the total-measure axioms, and for which the pure-state functional $\A^d_t(|\psi\rangle)$ is an entanglement monotone across the $t\,|\,N{-}t$ bipartition, the convex-roof extension $\A_t^{Q,d}$ satisfies axioms~\eqref{ax:Q1}--\eqref{ax:Q5}. If, in addition, $\A_t^{T,d}$ is concave on $\mathcal D(\mathcal H_j)$, then the ordering axiom~\eqref{ax:Q6} also holds. Under these assumptions, $\A_t^{Q,d}$ defines a quantum measure of $t$-AC associated with the corresponding total distance-based measure $\A_t^{T,d}$.
\end{proposition}

\subsubsection{Fidelity-based construction}
We now define the quantum measure associated with the total fidelity-based
quantity \eqref{BuresACTlin}. The key point is that, on pure states, the fidelity
appearing in \eqref{Fidelity.ACT} is directly related to the negativity across the
bipartition $t\,|\,N{-}t$. Applying Proposition~\ref{prop:FidelityNegativityRelation}
with $d_A=t+1$ gives, for every pure symmetric state $|\psi\rangle$,
\begin{equation}
F\big(\rho_t,\rho_0^{(t)}\big)
=
\frac{1+2\Nega_t(|\psi\rangle)}{t+1}.
\end{equation}
Consequently,
\begin{equation}
\A^{T,F}_t(|\psi\rangle)
=
\frac{2}{t}\,\Nega_t(|\psi\rangle).
\end{equation}
This suggests defining the associated quantum $t$-AC measure by the normalized
convex-roof negativity
\begin{equation}
  \A^{Q,F}_t(\rho)
  =
  \frac{2}{t}\,\Nega_t^{\mathrm{CR}}(\rho),
  \label{eq:BuresCompatibleQAC}
\end{equation}
where
\begin{equation}
  \Nega_t^{\mathrm{CR}}(\rho)
  =
  \min_{\{p_i,|\psi_i\rangle\}}
  \sum_i p_i\,\Nega_t(|\psi_i\rangle)
\end{equation}
is the convex-roof extension of the negativity across the $t\,|\,N{-}t$
bipartition~\cite{Lee2003}.

\begin{proposition}[Proof in Appendix \ref{Subsec.Prop7}]
\label{prop:Bures_quantum_Q1Q6}
The fidelity-based total candidate $\A^{T,F}_t$ and the quantum quantity $\A^{Q,F}_t$ defined in Eqs.~\eqref{BuresACTlin} and~\eqref{eq:BuresCompatibleQAC} satisfy the following properties:
The total candidate $\A^{T,F}_t$ satisfies axioms
\eqref{ax:T1}--\eqref{ax:T3}; numerical evidence suggests that it also satisfies axiom~\eqref{ax:T4}, but we leave this point open. The quantum quantity $\A^{Q,F}_t$ satisfies axioms \eqref{ax:Q1}, \eqref{ax:Q3}--\eqref{ax:Q5}, coincides with $\A^{T,F}_t$ on pure states as required by \eqref{ax:Q2}, and obeys the ordering axiom~\eqref{ax:Q6} with respect to this total candidate. In particular,
\begin{equation}
  \A_t^{Q,F}(\rho)\le \A_t^{T,F}(\rho)
  \qquad\forall\,\rho\in\mathcal D(\mathcal H_j).
\end{equation}
\end{proposition}

\paragraph*{Final remarks.}
The previous propositions show that, for the purity-based and
Hilbert--Schmidt distance-based constructions, the corresponding quantum measures
satisfy axioms \eqref{ax:Q1}--\eqref{ax:Q6} relative to the chosen total measure. For the fidelity-based construction, the quantum quantity
satisfies the quantum-side axioms and the ordering
$\mathcal A_t^{Q,F}\le \mathcal A_t^{T,F}$ relative to the total candidate
\eqref{BuresACTlin}, so the corresponding classical contribution
$\mathcal A_t^{C,F}=\mathcal A_t^{T,F}-\mathcal A_t^{Q,F}$ is nonnegative.
However, since axiom~\eqref{ax:T4} for $\mathcal A_t^{T,F}$ is supported here
only by numerical evidence, the fidelity-based pair should be regarded as a
well-motivated candidate total--quantum pair until a proof of total monotonicity
under all $\mathrm{SU}(2)$-covariant channels is obtained.

\subsection{Cumulative-multipole constructions}

Cumulative multipole measures introduced and studied in
\cite{SanchezSotoKlimovDeLaHozLeuchs2013,DeLaHozKlimovBjorkLeuchsSanchezSoto2013,BjorkKlimovDeLaHozLeuchsSanchezSoto2015,Goldberg2022,Goldberg2021}
provide a tensor-based characterization of the angular structure of spin states. Related work connects these multipoles with higher-order polarization coherences and quantum polarization structure~\cite{Goldberg2022,Goldberg2021}. In particular, they recover the standard distinction between first-order (Stokes-vector) polarization and higher-order, ``hidden'' angular structure by varying the cutoff rank: the case $t=1$ isolates first-order directional content, while $t\ge 2$ probes the presence (or suppression) of higher-order multipoles. In our setting, this same hierarchy fits naturally with mixed-state total anticoherence, and could be further refined by separating, for each $t$, the total isotropy into quantum and classical contributions.

From the irreducible tensor-operator expansion \eqref{rhoTLMexpansion}, define the rank-$L$ multipole weight and its cumulative version by
\begin{equation}
 C_{\le t}(\rho)=\sum_{L=1}^t \sum_{M=-L}^L |\rho_{LM}|^2.
\end{equation}
One has $C_{\le t}(\rho)=0$ iff $\rho$ is $t$-AC, i.e.\ all multipoles up to rank $t$ vanish. Moreover, the multipole weights $\sum_{M=-L}^L |\rho_{LM}|^2$ can be expressed as linear combinations of the reduced-state purities $\Tr(\rho_k^2)$ (see~\cite{Mar.Wei.Gir:20}), which makes explicit the link with the purity-based constructions.

\subsubsection{Total cumulative-multipole AC}

For $1\le t\le 2j-1$, a natural \emph{total} $t$-anticoherence measure is obtained by an affine decreasing rescaling of $C_{\le t}$ that vanishes on pure coherent states. For spin-coherent states $\rho_{\rm coh}=|\theta,\phi\rangle\langle\theta,\phi|$, one finds (independently of $\theta,\phi$)
\begin{equation}
  C_{\le t}(\rho_{\rm coh})
  =\frac{2 j}{2 j+1}-\frac{((2 j)!)^2}{(2 j-t-1)! (2 j+t+1)!}.
\end{equation}
In this subsection, we therefore restrict ourselves to $1\le t\le 2j-1$ and define
\begin{equation}\label{totalACcm}
  \mathcal A^{T,\mathrm{cm}}_t(\rho)
  = 1-\frac{C_{\le t}(\rho)}{C_{\le t}(\rho_{\rm coh})}.
\end{equation}

\begin{proposition}[Proof in Appendix \ref{Subsec.Prop8}]
\label{prop:cm_total}
For $1\le t\le 2j-1$, the quantity $\mathcal A^{T,\mathrm{cm}}_t$ satisfies axioms \eqref{ax:T1}--\eqref{ax:T4}, and therefore defines a valid measure of total $t$-anticoherence.
\end{proposition} 

\subsubsection{Quantum cumulative-multipole AC ?}

A natural attempt at defining a corresponding \emph{quantum} cumulative-multipole measure is to take the convex roof of the pure-state functional:
\begin{equation}
  \mathcal A^{Q,\mathrm{cm}}_t(\rho)
  =\min_{\{p_i,|\psi_i\rangle\}}
  \sum_i p_i\, \mathcal A^{T,\mathrm{cm}}_t(|\psi_i\rangle\langle\psi_i|).
  \label{CRcm}
\end{equation}
However, even though $\mathcal A^{Q,\mathrm{cm}}_t$ satisfies axioms
\eqref{ax:Q1}--\eqref{ax:Q4} and \eqref{ax:Q6}, it fails to satisfy the axiom \eqref{ax:Q5} about monotonicity under LOCC operations for $t\geq 2$, and therefore does not define a valid quantum AC measure. This already occurs at the pure-state level as the following example shows.

Consider the spin-$2$ pure states
\begin{align}
|\psi\rangle&=\frac{1}{\sqrt3}\Bigl(|2,2\rangle+|2,0\rangle+|2,-2\rangle\Bigr),\\
|\phi\rangle&=\frac{1}{\sqrt6}\Bigl(-2|2,2\rangle-|2,0\rangle+|2,-2\rangle\Bigr).
\end{align}
Across the bipartition $2|2$, their Schmidt spectra are
\begin{align}
\lambda(\psi)&=\left(\frac{7}{18}+\frac{\sqrt6}{9},\ \frac29,\ \frac{7}{18}-\frac{\sqrt6}{9}\right),\\
\lambda(\phi)&=\left(\frac49+\frac{\sqrt{87}}{36},\ \frac49-\frac{\sqrt{87}}{36},\ \frac19\right).
\end{align}
These satisfy $\lambda(\psi)\prec \lambda(\phi)$. Hence, by Nielsen's theorem~\cite{Nielsen1999}, there exists a deterministic LOCC transformation $|\psi\rangle\longrightarrow|\phi\rangle$
across the split $2|2$.

We now evaluate the cumulative multipole measure directly from Eq.~\eqref{totalACcm} for $t=2$ and obtain
\begin{equation}
    \mathcal A^{T,\mathrm{cm}}_2(|\psi\rangle\langle\psi|)=\frac{7}{12},
\qquad
\mathcal A^{T,\mathrm{cm}}_2(|\phi\rangle\langle\phi|)=\frac{31}{48}.
\end{equation}
Since both states are pure, their total AC measure is equal to their quantum AC measure, and thus
\begin{equation}
\mathcal A^{Q,\mathrm{cm}}_2(|\psi\rangle\langle\psi|)=\frac{7}{12}<\frac{31}{48}=\mathcal A^{Q,\mathrm{cm}}_2(|\phi\rangle\langle\phi|),
\end{equation}
although $|\psi\rangle\to|\phi\rangle$ is achievable by deterministic LOCC. This shows that the quantum cumulative-multipole  measure violates axiom \eqref{ax:Q5} for $t=2$. Similarly, one can find examples of violations for $t>2$. It is only for $t=1$ that \eqref{ax:Q5} holds due to the following Proposition.

\begin{proposition}[Proof in Appendix \ref{Subsec.Prop10}]
\label{prop:cm_t1_equiv_purity} 
For $t=1$, the cumulative-multipole and purity-based constructions are affinely equivalent on pure states. Consequently, their convex-roof quantum extensions coincide:
\begin{equation}
  \mathcal A^{Q,\mathrm{cm}}_1(\rho)=\A^{Q,\mathcal{P}}_1(\rho)
  \qquad\forall\,\rho\in\mathcal D(\mathcal H_j),
\end{equation}
with the normalizations chosen above. In particular, $\mathcal A^{Q,\mathrm{cm}}_1$ satisfies axiom \eqref{ax:Q5}.
\end{proposition}

\medskip
\noindent

\section{Evaluation of measures}
\label{Sec:evaluation}

\subsubsection*{ Example 1. Rank-2 mixed states for \texorpdfstring{$j=2$}{j=2}}

As a first explicit illustration of the interplay between total, quantum, and
classical anticoherence, we consider a rank-2 mixed state of a spin-$2$ system,
corresponding to $N=2j=4$ qubits,
\begin{equation}\label{Eq:j2_rank2_mixed_state}
\rho(\lambda)=\lambda\ket{\psi_+}\!\bra{\psi_+}+(1-\lambda)\ket{\psi_-}\!\bra{\psi_-},
\end{equation}
where the orthonormal eigenvectors are given by
\begin{equation}
\label{Eq:j2_rank2_states}
\ket{\psi_\pm}
= \tfrac{1}{2}
\bigl(
\ket{2,2}
\pm i\sqrt{2}\,\ket{2,0}
+ \ket{2,-2}
\bigr).
\end{equation}

\paragraph*{Fidelity-based measure for $t=1$.}
For arbitrary weight $\lambda$, the partial transposes
$(\ket{\psi_\pm}\!\bra{\psi_\pm})^{T_1}$ have supports on orthogonal subspaces.
As a consequence, the negativity is additive on the spectral decomposition of
$\rho$, i.e.,
\begin{equation}
\label{Eq:Neg1}
  \mathcal{N}_1(\rho)
  = \lambda \mathcal{N}_1(\ket{\psi_+})
  + (1-\lambda) \mathcal{N}_1(\ket{\psi_-}).
\end{equation}
Since each eigenvector is maximally $1$-AC,
$\mathcal{N}_1(\ket{\psi_\pm})=1/2$, it follows immediately that
\begin{equation}
  \mathcal{N}^{\mathrm{CR}}_1(\rho)
  = \mathcal{N}_1(\rho)
  = \tfrac{1}{2},
\end{equation}
independently of the mixing parameter $\lambda$.
Using the definition of the fidelity-based quantum measure then yields
\begin{equation}
  \mathcal{A}^{Q,F}_1(\rho)=1.
\end{equation}
Since the total measure also equals unity for $t=1$, one has
\begin{equation}
  \mathcal{A}^{T,F}_1(\rho)=\mathcal{A}^{Q,F}_1(\rho)=1,
  \qquad
  \mathcal{A}^{C,F}_1(\rho)=0.
\end{equation}
This provides an explicit example of a \emph{mixed} state that is perfectly
$1$-AC with a purely quantum origin. Such states have appeared previously in the literature.
In particular, the vectors~\eqref{Eq:j2_rank2_states} span a
\emph{$1$-AC subspace}, meaning that every pure state in their span is
$1$-AC \cite{Per.Pau:17,PhysRevA.111.022435}.

\paragraph*{Purity-based measure for $t=2$.}
We consider the same example as above. A direct evaluation of the reduced two-qubit state $\rho_2$ gives $ \Tr\rho_2^2=\frac{1}{3}$; hence, the total purity-based measure is $\A^{T,\mathcal{P}}_2(\rho)=1$,
independently of $\lambda$. We now determine the convex-roof quantum contribution. Since $\rho$ has
rank~$2$, we may identify its support with an effective qubit spanned by
$\{\ket{\psi_+},\ket{\psi_-}\}$ and write $\rho$ in Bloch form with
$\mathbf{r}=(0,0,r_z)$, $r_z=2\lambda-1$.
Any decomposition $\rho=\sum_i p_i\ket{\phi_i}\!\bra{\phi_i}$ corresponds
to an ensemble of pure Bloch vectors $\mathbf{n}_i$ with $\sum_i p_i \mathbf{n}_i=\mathbf{r}$, and therefore $\sum_i p_i z_i = r_z$, where $z_i$ denotes the $z$-component of $\mathbf{n}_i$.

For a pure state $\ket{\phi}$ in the support, the purity-based functional can be
written as
\begin{equation}
  \A^{\mathcal{P}}_2(\ket{\phi})=\frac{3+z^2}{4},
\end{equation}
Hence, for an arbitrary decomposition,
\begin{equation}
\begin{aligned}
    \sum_i p_i\,\A^{\mathcal{P}}_2(\ket{\phi_i})
  ={}& \frac{3}{4}+\frac{1}{4}\sum_i p_i z_i^2\\
  \ge{}& \frac{3}{4}+\frac{1}{4}\Bigl(\sum_i p_i z_i\Bigr)^2
  =\frac{3+r_z^2}{4},
\end{aligned}
  \label{Eq:CR_lowerbound}
\end{equation}
where we used Jensen's inequality for the convex function $z^2$.
The bound is tight. For $r_z=2\lambda-1$, choose $\theta$ such that $\cos\theta=r_z$
and define
\begin{equation}
    |\phi_{\pm}\rangle=\cos\frac{\theta}{2}\,|\psi_+\rangle \pm \sin\frac{\theta}{2}\,|\psi_-\rangle .
\end{equation}
Then $\rho=\frac12|\phi_+\rangle\langle\phi_+|+\frac12|\phi_-\rangle\langle\phi_-|$,
and both $|\phi_\pm\rangle$ have the same Bloch $z$-component $z_\pm=r_z$, so the
Jensen bound is saturated. Therefore,
\begin{equation}
  \A^{Q,\mathcal{P}}_2(\rho)
  =\frac{3+r_z^2}{4}
  =\frac{3+(2\lambda-1)^2}{4}
  =\lambda^2-\lambda+1.
\end{equation}
Finally, the corresponding classical contribution is
\begin{equation}
  \A^{C,\mathcal{P}}_2(\rho)
  =\A^{T,\mathcal{P}}_2(\rho)-\A^{Q,\mathcal{P}}_2(\rho)
  =\lambda(1-\lambda),
\end{equation}
which vanishes for $\lambda\in\{0,1\}$ and is maximal for the equal mixture
$\lambda=1/2$. This result is illustrated in Figure \ref{fig:ACMixed_example}.

\paragraph*{Purity-based measure for $t=3$.}
For $t=3$, an analogous calculation gives
\begin{equation}
  \mathcal{A}^{Q,\mathcal{P}}_3(\rho)=\tfrac{2}{3},
\end{equation}
which is independent of $\lambda$.
The total measure, however, now depends on the mixing parameter and reads
\begin{equation}
  \mathcal{A}^{T,\mathcal{P}}_3(\rho)
  = \tfrac{2}{3}\bigl(-2\lambda^2+2\lambda +1\bigr).
\end{equation}
The classical contribution is thus
\begin{equation}
  \mathcal{A}^{C,\mathcal{P}}_3(\rho)
  = \mathcal{A}^{T,\mathcal{P}}_3(\rho)-\mathcal{A}^{Q,\mathcal{P}}_3(\rho)
  = \tfrac{4}{3}\lambda(1-\lambda).
\end{equation}

\paragraph*{Fidelity-based measures for $t=2,3$.}
The fidelity-based measures can also be evaluated analytically for this family.
For $t=2$, the reduced state is maximally mixed for all $\lambda$, so
\begin{equation}
  \mathcal A^{T,F}_2(\rho)=1 .
\end{equation}
It remains to evaluate the convex-roof negativity.  Let
$z=2\lambda-1$ and identify the support of $\rho$ with a qubit.  For a pure state
$|\phi\rangle=\alpha|\psi_+\rangle+\beta|\psi_-\rangle$, define
\begin{equation}
 x=2\operatorname{Re}(\alpha^*\beta),\qquad
 y=2\operatorname{Im}(\alpha^*\beta),\qquad
 z_\phi=|\alpha|^2-|\beta|^2 .
\end{equation}
The spectrum of the two-qubit reduced state is
\begin{equation}
p_\pm=\frac13+\frac{x}{6}\pm\frac{\sqrt3\,y}{6},\quad
 p_3=\frac13-\frac{x}{3},
\end{equation}
and therefore
\begin{equation}
 \mathcal A^F_2(|\phi\rangle)
 =\mathcal N_2(|\phi\rangle)
 =\sqrt{p_+p_-}+\sqrt{p_+p_3}+\sqrt{p_-p_3}.
\end{equation}
An elementary minimization over $x,y$ at fixed
$z_\phi$ gives
\begin{equation}
  \mathcal A^F_2(|\phi\rangle)
  \ge \frac{1+|z_\phi|}{2} .
\end{equation}
Consequently, consider an arbitrary pure-state decomposition
$\rho(\lambda)=\sum_i p_i|\phi_i\rangle\langle\phi_i|$, with
$|\phi_i\rangle=\alpha_i|\psi_+\rangle+\beta_i|\psi_-\rangle$.
Since this decomposition must reproduce the Bloch vector of $\rho(\lambda)$
inside the support $\mathrm{span}\{|\psi_+\rangle,|\psi_-\rangle\}$, its
$z$-components satisfy
\begin{equation}
\sum_i p_i z_i=2\lambda-1,
\qquad
z_i=|\alpha_i|^2-|\beta_i|^2 .
\end{equation}
Using the bound above for each pure state in the decomposition, we obtain
\begin{equation}
\begin{aligned}
 \sum_i p_i\mathcal A^F_2(|\phi_i\rangle)
 &\ge \frac12+\frac12\sum_i p_i |z_i|  \\
 &\ge \frac12+\frac12\left|\sum_i p_i z_i\right| \\
 &=\frac12+\left|\lambda-\frac12\right| .
\end{aligned}
\end{equation}
This lower bound is tight.  For $\lambda\ge1/2$, write
\begin{equation}
 \rho(\lambda)=(2\lambda-1)|\psi_+\rangle\langle\psi_+|
 +2(1-\lambda)\rho_*,
\end{equation}
where
\begin{equation}
 \rho_* =\frac12\left(|\psi_+\rangle\langle\psi_+|
 +|\psi_-\rangle\langle\psi_-|\right)
 =\frac13\sum_{k=0}^{2}|\phi_k\rangle\langle\phi_k|
\end{equation}
with
\begin{equation}
 |\phi_k\rangle=\frac{1}{\sqrt2}
 \left(|\psi_+\rangle+e^{2\pi i k/3}|\psi_-\rangle\right).
\end{equation}
Since $\mathcal A^F_2(|\psi_+\rangle)=1$ and
$\mathcal A^F_2(|\phi_k\rangle)=1/2$, this decomposition has average value
$\lambda$.  The case $\lambda\le1/2$ is obtained by exchanging
$|\psi_+\rangle$ and $|\psi_-\rangle$.  Hence
\begin{equation}
  \mathcal A^{Q,F}_2(\rho)
  =\frac12+\left|\lambda-\frac12\right|,
  \qquad
  \mathcal A^{C,F}_2(\rho)
  =\frac12-\left|\lambda-\frac12\right| .
\end{equation}
For $t=3$, the reduced state $\rho_3$ has spectrum
\begin{equation}
 \left\{\frac{\lambda}{2},\frac{\lambda}{2},
 \frac{1-\lambda}{2},\frac{1-\lambda}{2}\right\},
\end{equation}
which gives
\begin{equation}
  \mathcal A^{T,F}_3(\rho)
  =\frac{1+4\sqrt{\lambda(1-\lambda)}}{3}.
\end{equation}
Moreover, every pure state in the support has the same negativity across the
$3|1$ bipartition, namely $\mathcal N_3=1/2$. Thus
\begin{equation}
  \mathcal A^{Q,F}_3(\rho)=\frac13,
  \qquad
  \mathcal A^{C,F}_3(\rho)=\frac{4}{3}\sqrt{\lambda(1-\lambda)} .
\end{equation}
Together with the purity-based expressions above, these analytical results are
shown in Fig.~\ref{fig:ACMixed_example}.

This example clearly illustrates how increasing the order $t$ enhances the
relative weight of the classical contribution to isotropy.
While $1$-anticoherence is entirely quantum and robust under mixing, higher-order
anticoherence increasingly relies on probabilistic mixing, even though the total
isotropy may remain high.

\subsection{States saturating the convex-roof negativity}
\label{Sec:NegSaturation}
For the quantum AC measure associated with the fidelity-based construction, one must calculate the convex-roof extension of the negativity $\mathcal{N}^{\mathrm{CR}}_t(\rho)$. Here, we discuss two cases where this quantity can be computed explicitly.
\subsubsection{Mixed states with support of equal negativity}
The first case is when every state in the support of $\rho$ has the same negativity. In this case, it is immediate to obtain that for any pure decomposition of $\rho= \sum_{i}p_i \ket{\psi_i}\bra{\psi_i}$, $\sum_i p_i \Nega (\ket{\psi_i})$ has the same value. Consequently, there is no need to do a minimization to evaluate $\Nega_t^{\mathrm{CR}}(\rho)$. In particular, one can directly calculate it from the eigenbasis of the state $\Nega_t^{\mathrm{CR}}(\rho) = \sum_{i} \lambda_i \Nega(\ket{\phi_i})$. 

This situation occurs for mixed states whose support lies in a $t$-AC subspace. By definition, any state in a $t$-AC subspace is $t$-AC~\cite{Per.Pau:17,PhysRevA.111.022435}. Moreover, a $t$-AC pure state has maximum negativity $\Nega_t (\ket{\psi} ) =t/2$ ($t\le N/2$)~\cite{PhysRevA.111.022435,Den.Mar:22}. 
Whenever the support of a mixed state is itself a $t$-AC subspace,
the resulting state displays maximal quantum $t$-anticoherence, independently
of the specific mixing probabilities.

\begin{theorem}
\label{Thm:ACSubspace}
Let $\rho$ be a mixed spin-$j$ state whose image\footnote{The image of an operator is the subspace spanned by its eigenvectors with nonzero eigenvalues.}
$\mathrm{Im}(\rho)\subset\mathcal H_j$ is a $t$-AC subspace. Then every pure
state in every decomposition of $\rho$ belongs to this subspace and is
$t$-AC. Hence each such pure state has maximal negativity
$\mathcal{N}_t=t/2$ across the $t\,|\,N{-}t$ bipartition, and the
convex-roof negativity is
\begin{equation}
\mathcal{N}_t^{\mathrm{CR}}(\rho)=t/2.
\end{equation}
Consequently, the associated quantum $t$-anticoherence measure attains
its maximal value,
$\mathcal{A}^{Q,F}_t(\rho)=1.$
If the stronger ordinary-negativity equality
$\mathcal N_t(\rho)=t/2$ holds for a particular $t$-AC subspace, as in the
settings characterized in Ref.~\cite{PhysRevA.111.022435}, then the ordinary
negativity and the convex-roof negativity coincide for such states.
\end{theorem}

\noindent
Since the total $t$-anticoherence is also maximal for such states, the classical
contribution necessarily vanishes.
Therefore, any mixed state supported on a $t$-AC subspace exhibits
\emph{purely quantum} $t$-anticoherence, extending to mixed states a property previously known for pure AC states.

\subsubsection*{Example 2. \texorpdfstring{$2$}{2}-AC subspaces for \texorpdfstring{$N=7$}{N=7}}

As a second example, consider the symmetric $N=7$ qubit states~\cite{Gross_2021,lim_2023}
\begin{equation}
\label{Eq:j7_AC_subspace}
\begin{aligned}
\ket{\psi_1} &=
\sqrt{\tfrac{3}{10}}
\ket{\tfrac{7}{2},\tfrac{7}{2}}
+
\sqrt{\tfrac{7}{10}}
\ket{\tfrac{7}{2},-\tfrac{3}{2}},\\
\ket{\psi_2} &=
\sqrt{\tfrac{7}{10}}
\ket{\tfrac{7}{2},\tfrac{3}{2}}
-
\sqrt{\tfrac{3}{10}}
\ket{\tfrac{7}{2},-\tfrac{7}{2}} .
\end{aligned}
\end{equation}
These vectors span a $2$-AC subspace.
Therefore, for any mixture $\rho$ supported on this subspace one finds
\begin{equation}
  \mathcal{N}^{\mathrm{CR}}_1(\rho)=\tfrac{1}{2},
  \qquad
  \mathcal{N}^{\mathrm{CR}}_2(\rho)=1,
\end{equation}
and hence
$\mathcal{A}^{T,F}_t(\rho)=\mathcal{A}^{Q,F}_t(\rho)=1$
for $t=1,2$.
\subsubsection{Pairwise orthogonal partial transposed projectors}
By definition, the negativity satisfies 
\begin{equation}
\mathcal{N}^{\mathrm{CR}}_t(\rho)\ge \mathcal{N}_t(\rho) .
\end{equation}
A sufficient condition for saturation is the following:
consider a decomposition of the state
$\rho=\sum_i p_i \ket{\psi_i}\!\bra{\psi_i}.$
If the images of the partially transposed projectors $
(\ket{\psi_i}\!\bra{\psi_i})^{T_A}$ are pairwise orthogonal subspaces\footnote{Two subspaces are orthogonal if every vector in one is orthogonal to every vector in the other.}, then the negativity is additive over the decomposition\footnote{This follows from the following result: let $A$ and $B$ be two diagonalizable matrices, and let $\lambda_{\neq 0}(A)$ be the set of nonzero eigenvalues (counting multiplicity) of $A$. Then, $\lambda_{\neq 0}(A+B) = \lambda_{\neq 0}(A) \cup \lambda_{\neq 0}(B) $ if $\mathrm{im}(A) \perp \mathrm{im}(B)$. In particular, when this holds, the negative eigenvalues of $A+B$ are the union of the negative eigenvalues of $A$ and $B$.},
\begin{equation}
\label{Eq:NegAdditivity}
  \sum_i p_i\,\mathcal{N}_t(\ket{\psi_i})
  = \mathcal{N}_t(\rho).
\end{equation}
In this case, the given decomposition is optimal for the convex-roof construction. Indeed, the given decomposition yields
\begin{equation}
    \mathcal N_t(\rho) \leq \mathcal N_t^{\mathrm{CR}}(\rho)\le
\sum_i p_i\mathcal N_t(\ket{\psi_i})
= \mathcal N_t(\rho),
\end{equation}
so equality holds throughout, and the convex roof is easily evaluated. Such orthogonality among the partially transposed operators also appears in the orthogonal basis of an AC subspace (see Appendix F in Ref.~\cite{PhysRevA.111.022435}). However, they are not the only cases. Another nontrivial example that does not arise from an AC subspace is given by two spin-$5/2$ states:
\begin{equation}
\begin{aligned}
\label{Eq.Ex5half}
\ket{\psi_1}
&= \tfrac{1}{\sqrt{6}}
\ket{\tfrac{5}{2},-\tfrac{5}{2}}
+ \sqrt{\tfrac{5}{6}}
\ket{\tfrac{5}{2},\tfrac{3}{2}},\\
\ket{\psi_2}
&= -\sqrt{\tfrac{5}{6}}
\ket{\tfrac{5}{2},-\tfrac{3}{2}}
+ \tfrac{1}{\sqrt{6}}
\ket{\tfrac{5}{2},\tfrac{5}{2}} .
\end{aligned}
\end{equation}
In their representation as $N=5$ symmetric qubits, and considering the partial transposed density matrices in the bipartition $t|N-t$ with $t=1$, we obtain that $(\ket{\psi_1}\bra{\psi_1})^{T_A} \perp (\ket{\psi_2}\bra{\psi_2})^{T_A} $. Therefore,
\begin{multline*}
\Nega_1^{\mathrm{CR}} \big(\lambda \ket{\psi_1}\bra{\psi_1} + (1-\lambda) \ket{\psi_2}\bra{\psi_2} \big) 
\\
= \lambda\, \Nega_1 \big(\ket{\psi_1} \big) + (1-\lambda)\, \Nega_1 \big(\ket{\psi_2} \big) = \frac{\sqrt{2}}{3} .  
\end{multline*}

\begin{figure}[t]
\centering
\includegraphics[width=0.85\linewidth]{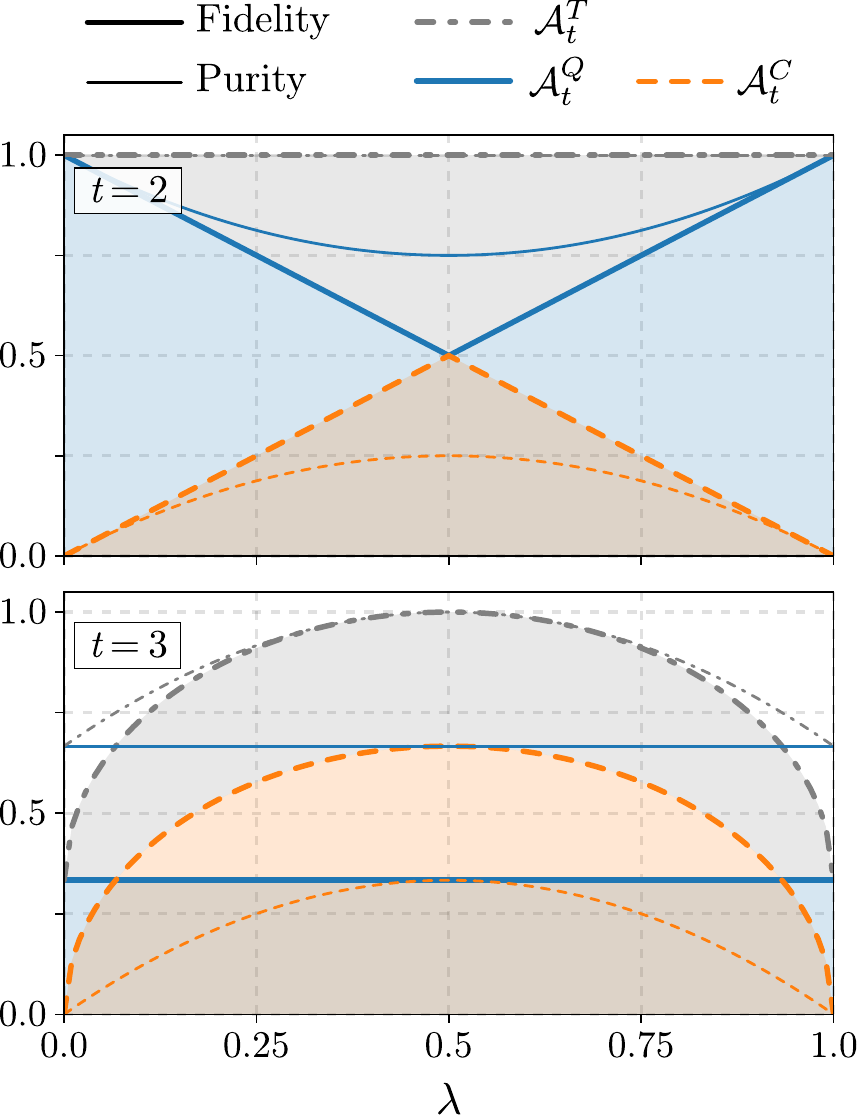}
\caption{
Total, quantum, and classical $t$-anticoherence of the mixed state~\eqref{Eq:j2_rank2_mixed_state} as functions of the mixing weight $\lambda$. 
The upper and lower panels correspond to $t=2$ and $t=3$, respectively. 
Thick curves correspond to the fidelity-based definitions
\eqref{BuresACTlin} and \eqref{eq:BuresCompatibleQAC}, evaluated analytically in Sec.~\ref{Sec:evaluation}; thin curves show the purity-based measures.  In particular,
$\mathcal A^{Q,F}_2=\frac12+|\lambda-\frac12|$ and
$\mathcal A^{Q,F}_3=\frac13$.
For $t=1$, both measures give $\mathcal{A}^{T}_1=\mathcal{A}^{Q}_1=1$ and $\mathcal{A}^{C}_1=0$ (data not shown).
}
\label{fig:ACMixed_example}
\end{figure}
\section{Applications}
\label{Sec:Applications}
\subsection{Robustness of quantum AC under particle loss}

Figure~\ref{fig:QACreduced} compares the robustness of the quantum, purity-based
$t$-anticoherence $\mathcal{A}^{Q,\mathcal{P}}_t$ under particle loss for three emblematic
families of symmetric $N$-qubit pure states: highest-order anticoherent pure (HOAP) states~\cite{Den.Mar:22}, $W$ states, and GHZ states. For each $N$ and each remaining subsystem size $q$, we consider the reduced state $\rho_q=\Tr_{N-q}(|\psi\rangle\!\langle\psi|)$ and evaluate $\mathcal{A}^{Q,\mathcal{P}}_t(\rho_q)$ for $t=1$ (left column) and $t=2$ (right column)~\cite{Zhu2025}.

Two qualitative behaviors stand out. First, HOAP states, which are AC of order $t\geq 2$ for $N=4$ and $N\geq 6$, retain a large \emph{quantum} anticoherence for moderate particle losses: $\mathcal{A}^{Q,\mathcal{P}}_t(\rho_q)$ remains high as long as a significant fraction of the particles is conserved and only becomes negligible once approximately half (or more) of the system is traced out. In this regime, the anticoherence of reduced states becomes mainly classical, in the sense that any remaining lower-order isotropy is not supported by quantum correlations across the relevant bipartitions. Second, this behavior contrasts sharply with $1$-AC GHZ states: their quantum anticoherence is extremely fragile, collapsing as soon as any particle is traced
out. In Fig.~\ref{fig:QACreduced}, this appears as nonzero values only on the ``no-loss'' diagonal $q=N$, while reduced states with $q<N$ exhibit
$\mathcal{A}^{Q,\mathcal{P}}_t = 0$. This is consistent with the well-known fact that particle loss destroys the entanglement structure of GHZ states, rendering the reduced states separable. Finally, $W$ states, which are not exactly $t$-AC in the orders considered here, show the most persistent quantum anticoherence under loss: $\mathcal{A}^{Q,\mathcal{P}}_t(\rho_q)$ remains visibly nonzero over a broad range of $q$, even when many particles are discarded. This robustness comes at the price of not reaching maximal quantum anticoherence: compared with HOAP states, the values are typically smaller (moderate marker sizes and intermediate colors), indicating a more uniform but partial quantum contribution to isotropy.

\begin{figure}[t]
\centering
\includegraphics[width=\linewidth]{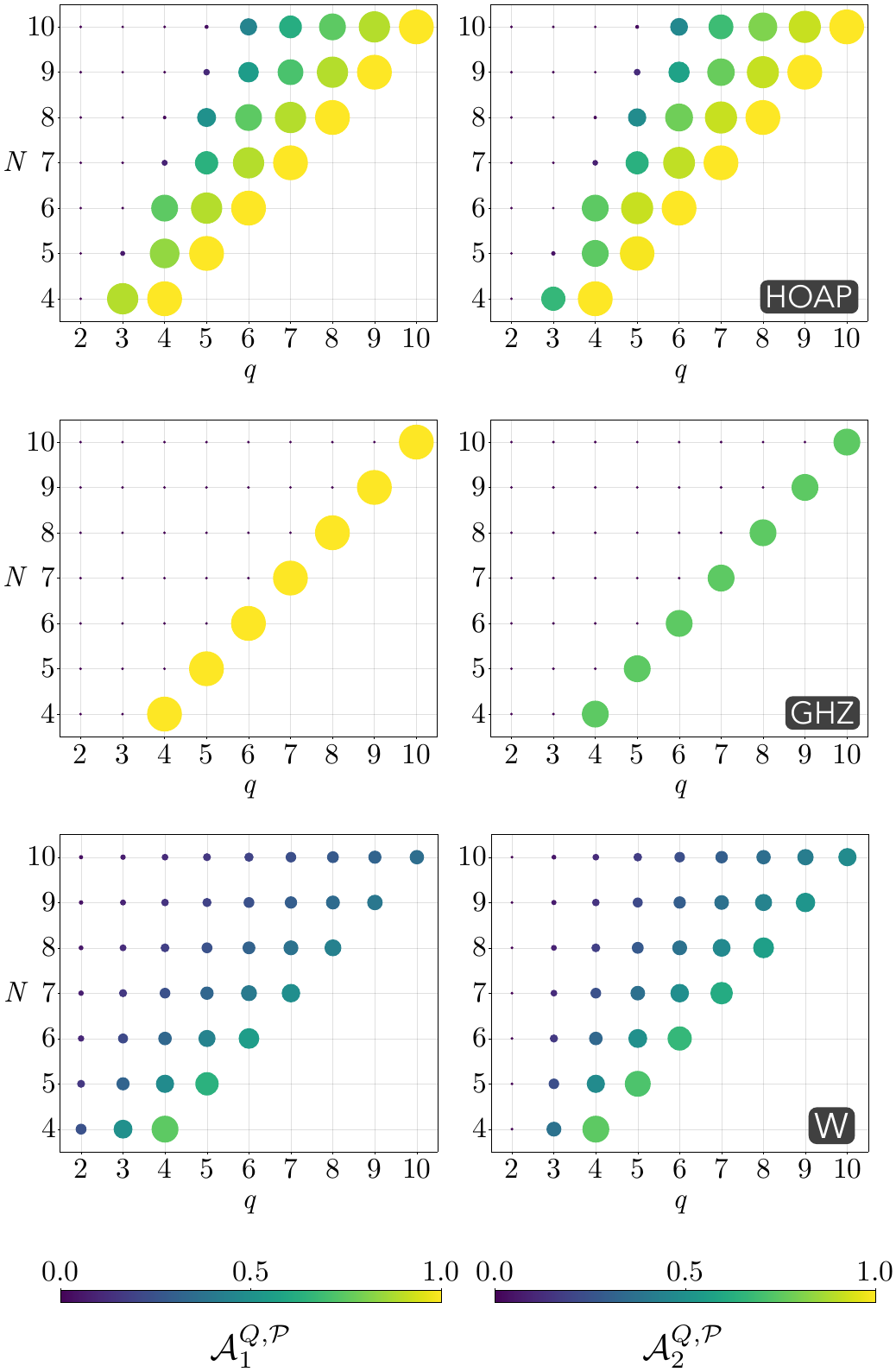}
\caption{Quantum purity-based $t$-anticoherence of reduced $q$-qubit states (seen as spin-$q/2$ states). Left (right): $t=1$ ($t=2$). Reduced states are obtained by partial tracing HOAP states~\cite{Den.Mar:22} (top), GHZ states (middle), and $W$ states (bottom). Disk areas scale with the magnitude of $\mathcal{A}^{Q,\mathcal{P}}_t$; dark-purple disks of minimal size indicate vanishing anticoherence.}
\label{fig:QACreduced}
\end{figure}

\subsection{Maximal anticoherence order for mixed states at fixed purity}
\label{Sec:MaxOrderPurity}

For mixed states, anticoherence competes with purity: increasing isotropy generally requires mixing, while higher purity constrains the degree to which low-order spin moments can be suppressed. A natural question is therefore the following: \emph{given a fixed global purity, what is the maximal order $t$ of anticoherence that a mixed spin state can achieve?}

This question is particularly relevant from the perspective of quantum resources. At fixed purity, different states may display the same total amount of isotropy while differing markedly in the relative contributions arising from genuinely quantum correlations and from statistical mixing. Our framework allows one to resolve this distinction by separating total, quantum, and classical contributions to $t$-anticoherence.

For non-symmetric multipartite systems, related trade-offs between purity and
isotropy have been investigated in Ref.~\cite{Kobus2019}. Here, we use the same type of techniques based on semidefinite programming but adapted to symmetric states and characterise the maximum achievable anticoherence order as a function of purity within this physically relevant sector.

Figure~\ref{fig:MaxorderACMixed} summarizes our results for $j=3/2, 2$ and $5/2$. For each value of the purity, we identify mixed states that are exactly $t$-AC, and hence maximize any admissible total $t$-AC measure by axiom \eqref{ax:T1}; we then indicate the highest such order $t$. The accompanying pie charts display the relative contributions of quantum ($\mathcal{A}_t^{Q,\mathcal{P}}$) and classical ($\mathcal{A}_t^{C,\mathcal{P}}$) anticoherence at that order. A completely filled chart corresponds to a perfectly $t$-AC state. The highest achievable order $t$ is piecewise constant: as $\Tr(\rho^2)$ decreases, successive thresholds are crossed where an additional multipolar rank can be completely removed, and the maximum achievable order increases by one step (blue staircase curve). To achieve higher orders of total anticoherence, the mixture must be increased: total isotropy at high $t$ is only accessible when the purity is sufficiently low. Along this staircase, the \emph{composition} of the isotropy changes: the charts indicate that the quantum contribution $\mathcal{A}_t^{Q,\mathcal{P}}$ (blue) generally decreases as $t$ increases, while the complementary classical contribution $\mathcal{A}_t^{C,\mathcal{P}}$ (orange) increases. In other words, reducing the purity allows for stronger total isotropy, but an increasing fraction of this isotropy comes from statistical mixing rather than genuine quantum correlations.

We complement Fig.~\ref{fig:MaxorderACMixed} with explicit mixed states realizing the different steps of the purity--AC trade-off. These examples serve two purposes. First, they show that the bounds are not merely numerical: the extremal points can be attained by simple low-rank mixtures in the Dicke basis. Second, they illustrate how the nature of isotropy changes with the order $t$. For each state and each relevant order, Table~\ref{tab:anticoherence_summary} reports the triple
$\bigl(
    \mathcal{A}_t^{T,\mathcal{P}},
    \mathcal{A}_t^{Q,\mathcal{P}},
    \mathcal{A}_t^{C,\mathcal{P}}
  \bigr)$, 
together with the rank and the global purity. The general trend is that reaching higher orders of total AC at fixed spin requires higher rank and more mixing, and the additional isotropy is increasingly classical. By contrast, a large quantum contribution is possible only at lower orders, or in the special case where the support of the mixed state forms a $t$-AC subspace.

The staircase in Fig.~\ref{fig:MaxorderACMixed} should be interpreted as follows. The displayed jumps are certified by explicit feasible states given below and by semidefinite-programming searches in the symmetric sector. In these searches the $t$-AC constraints are imposed as the linear constraints $\rho_t=\mathbb 1_{t+1}/(t+1)$, equivalently by the vanishing of the multipoles with $1\le L\le t$. Since the global purity constraint $\Tr(\rho^2)= \mathrm{const.}$ is quadratic, the numerical search is performed by optimizing linear objectives over fixed-rank or fixed-support ansatz families and by checking the achieved purities a posteriori. Thus the explicit states prove attainability of the indicated steps.

\begin{figure*}[t]
\centering
\includegraphics[width=0.975\linewidth]{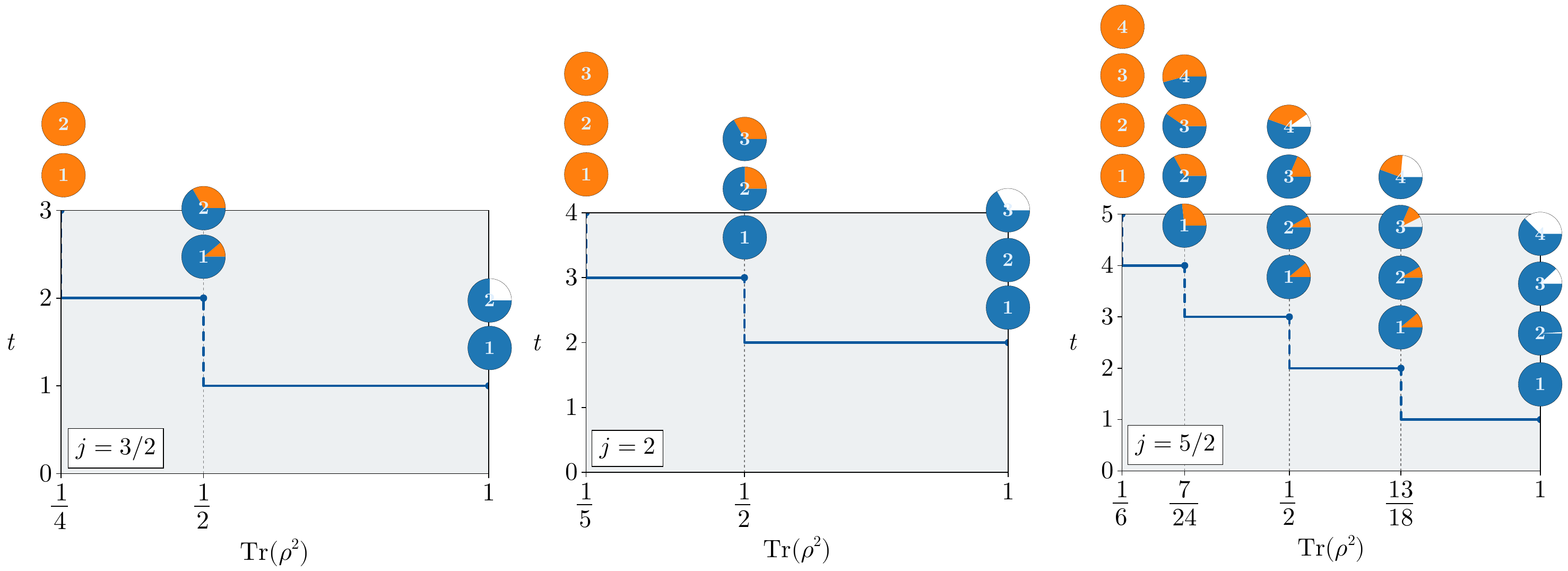}
\caption{Maximal achievable total AC order as a function of the global purity. 
For each purity threshold, the pie chart shows the decomposition of the purity-based AC measure into its quantum contribution $\A^{Q,\mathcal{P}}_t$ (blue) and classical contribution $\A^{C,\mathcal{P}}_t$ (orange), evaluated at the maximal attainable order $t$. The number at the center of each pie gives this order. A fully filled pie indicates a state that is exactly $t$-AC. The states realizing the jumps are given in Eqs.~\eqref{eq:j3o2_2AC}, \eqref{eq:j2_3AC}, \eqref{eq:j5o2_2AC}, \eqref{eq:j5o2_3AC}, and~\eqref{eq:j5o2_4AC}; see also Table~\ref{tab:anticoherence_summary}. The leftmost charts correspond to the MMS state, while the rightmost charts correspond to HOAP states~\cite{Den.Mar:22}.}
\label{fig:MaxorderACMixed}
\end{figure*}

\subsubsection{\texorpdfstring{$j=\tfrac{3}{2}$}{j=3/2}}

The first nontrivial example occurs for $j=3/2$. For pure spin-$3/2$ states, the maximal anticoherence order is only $1$; reaching second-order isotropy therefore requires mixing. Consider the two orthonormal states
\begin{equation}
\begin{aligned}
\ket{\psi_1}
&= \tfrac{1}{\sqrt{2}}
\bigl(
\ket{\tfrac{3}{2},\tfrac{3}{2}}
+
\ket{\tfrac{3}{2},-\tfrac{1}{2}}
\bigr),\\
\ket{\psi_2}
&= \tfrac{1}{\sqrt{2}}
\bigl(
-\ket{\tfrac{3}{2},\tfrac{1}{2}}
+
\ket{\tfrac{3}{2},-\tfrac{3}{2}}
\bigr).
\end{aligned}
\end{equation}
Their equal mixture,
\begin{equation}\label{eq:j3o2_2AC}
\rho
= \tfrac{1}{2}
\bigl(
\ket{\psi_1}\!\bra{\psi_1}
+
\ket{\psi_2}\!\bra{\psi_2}
\bigr),
\end{equation}
has purity $1/2$ and is exactly $2$-AC. It thus provides the lowest-spin example in which mixing enables a higher anticoherence order than is possible for pure states.

\subsubsection{\texorpdfstring{$j=2$}{j=2}}

For $j=2$, a rank-2 construction already reaches one order higher. Define
\begin{equation}
\ket{\psi_\pm}
= \tfrac{1}{2}
\bigl(
\ket{2,2}
\pm i\sqrt{2}\,\ket{2,0}
+ \ket{2,-2}
\bigr).
\end{equation}
The equal mixture
\begin{equation}\label{eq:j2_3AC}
\rho
= \tfrac{1}{2}
\bigl(
\ket{\psi_+}\!\bra{\psi_+}
+
\ket{\psi_-}\!\bra{\psi_-}
\bigr)
\end{equation}
has purity $1/2$ and is exactly $3$-AC. This example is particularly useful because the same two-dimensional support is also a $1$-AC subspace, so first-order isotropy remains entirely quantum, while higher-order isotropy contains an increasing classical contribution (see the pie charts in the middle panel of Fig.~\ref{fig:MaxorderACMixed}).

\subsubsection{\texorpdfstring{$j=\tfrac{5}{2}$}{j=5/2}}

For $j=5/2$, several extremal regimes can be realized depending on the rank and purity of the mixture (see Fig.~\ref{fig:MaxorderACMixed}, right panel).

\paragraph*{Rank-2 states.}
Let
\begin{equation}
\begin{aligned}
\ket{\psi_1}
&= \tfrac{1}{\sqrt{2}}
\bigl(
\ket{\tfrac{5}{2},\tfrac{5}{2}}
+
\ket{\tfrac{5}{2},-\tfrac{5}{2}}
\bigr),\\
\ket{\psi_2}
&= \tfrac{1}{\sqrt{2}}
\bigl(
\ket{\tfrac{5}{2},\tfrac{3}{2}}
+
\ket{\tfrac{5}{2},-\tfrac{3}{2}}
\bigr).
\end{aligned}
\end{equation}
The unequal mixture
\begin{equation}\label{eq:j5o2_2AC}
\rho
= \tfrac{1}{6}\ket{\psi_1}\!\bra{\psi_1}
+ \tfrac{5}{6}\ket{\psi_2}\!\bra{\psi_2}
\end{equation}
has purity $13/18$ and is exactly $2$-AC. Thus, at relatively high purity, only low-order isotropy can be made perfect.

A lower-purity rank-2 mixture reaches third-order anticoherence. Let $\ket{\psi_1}$ and $\ket{\psi_2}$ be the states of Eq.~\eqref{Eq.Ex5half}.
Then
\begin{equation}\label{eq:j5o2_3AC}
\rho
= \tfrac{1}{2}
\bigl(
\ket{\psi_1}\!\bra{\psi_1}
+
\ket{\psi_2}\!\bra{\psi_2}
\bigr)
\end{equation}
has purity $1/2$ and is exactly $3$-AC. Compared with the previous state, the additional order of isotropy is obtained by increasing the degree of mixing.

\paragraph*{Rank-4 state.}
Finally, fourth-order anticoherence requires a still more mixed state. Let
\begin{equation}
\begin{aligned}
\ket{\psi_1} &= \ket{\tfrac{5}{2},-\tfrac{1}{2}}, \\
\ket{\psi_2} &= \ket{\tfrac{5}{2},-\tfrac{3}{2}}, \\
\ket{\psi_3} &= \sqrt{\tfrac{9}{20}}
\ket{\tfrac{5}{2},\tfrac{5}{2}}
+ \sqrt{\tfrac{11}{20}}
\ket{\tfrac{5}{2},-\tfrac{5}{2}}, \\
\ket{\psi_4} &= \ket{\tfrac{5}{2},\tfrac{1}{2}} .
\end{aligned}
\end{equation}
The rank-4 mixture
\begin{equation}\label{eq:j5o2_4AC}
\rho
= \tfrac{1}{12}\ket{\psi_1}\!\bra{\psi_1}
+ \tfrac{1}{4}\ket{\psi_2}\!\bra{\psi_2}
+ \tfrac{1}{3}
\bigl(
\ket{\psi_3}\!\bra{\psi_3}
+ \ket{\psi_4}\!\bra{\psi_4}
\bigr)
\end{equation}
has purity $7/24$ and is exactly $4$-AC. This example illustrates the general pattern most clearly: the highest anticoherence order is achieved only after substantial mixing, and the corresponding total isotropy is dominated by its classical component.

Table~\ref{tab:anticoherence_summary} summarizes the AC measures of the representative states discussed in Sec.~\ref{Sec:MaxOrderPurity}.

\begin{table*}[t]
\centering
\renewcommand{\arraystretch}{1.5}
\begin{tabular}{c@{\hspace{1.5em}}c@{\hspace{1.5em}}c@{\hspace{1.5em}}c@{\hspace{1.5em}} c@{\hspace{1.5em}} c}
\hline
$j$ & State & Rank & Purity & $t$ & $\bigl(\mathcal{A}_{t}^{T,\mathcal{P}},\,\mathcal{A}_{t}^{Q,\mathcal{P}},\,\mathcal{A}_{t}^{C,\mathcal{P}}\bigr)$ \\ \hline
$\tfrac{3}{2}$ & \eqref{eq:j3o2_2AC} & 2 & $\tfrac{1}{2}$ & 1,2 & $(1,\tfrac{8}{9},\tfrac{1}{9})$, $(1,\tfrac{2}{3},\tfrac{1}{3})$ \\ \hline
$2$ & \eqref{eq:j2_3AC} & 2 & $\tfrac{1}{2}$ & 1,2,3 & $(1,1,0)$, $(1,\tfrac{3}{4},\tfrac{1}{4})$, $(1,\tfrac{2}{3},\tfrac{1}{3})$ \\ \hline
$\tfrac{5}{2}$ & \eqref{eq:j5o2_2AC} & 2 & $\tfrac{13}{18}$ & 1,2,3,4 &
$(1,\tfrac{8}{9},\tfrac{1}{9})$,
$(1,\tfrac{11}{12},\tfrac{1}{12})$,
$(\tfrac{25}{27},\tfrac{22}{27},\tfrac{1}{9})$,
$(\tfrac{55}{72},\tfrac{5}{9},\tfrac{5}{24})$ \\ \hline
$\tfrac{5}{2}$ & \eqref{eq:j5o2_3AC} & 2 & $\tfrac{1}{2}$ & 1,2,3,4 &
$(1,\tfrac{8}{9},\tfrac{1}{9})$,
$(1,\tfrac{11}{12},\tfrac{1}{12})$,
$(1,\tfrac{22}{27},\tfrac{5}{27})$,
$(\tfrac{65}{72},\tfrac{5}{9},\tfrac{25}{72})$ \\ \hline
$\tfrac{5}{2}$ & \eqref{eq:j5o2_4AC} & 4 & $\tfrac{7}{24}$ & 1,2,3,4 &
$(1.0,0.732,0.268)$,
$(1.0,0.671,0.329)$,
$(1.0,0.596,0.404)$,
$(1.0,0.458,0.542)$ \\
\hline
\end{tabular}
\caption{
Summary of mixed spin states discussed in Sec.~\ref{Sec:MaxOrderPurity}. For each state, the table lists the spin $j$, the equation defining the state, its rank, its purity $\Tr(\rho^2)$, and the AC orders of the evaluated measures in the  last column. The purity-based total, quantum, and classical contributions
$\bigl(\mathcal{A}_{t}^{T,\mathcal{P}},\,\mathcal{A}_{t}^{Q,\mathcal{P}},\,\mathcal{A}_{t}^{C,\mathcal{P}}\bigr)$ are given. Decimal entries are numerical convex-roof evaluations. States with $\mathcal{A}_{t}^{T,\mathcal{P}}=1$ are exactly $t$-anticoherent.
}
\label{tab:anticoherence_summary}
\end{table*}
%
\section{Conclusions}
\label{Sec.Conc}

We have developed a systematic framework to quantify spin anticoherence beyond the pure-state setting. In its original formulation, spin $t$-anticoherence is a property of maximal isotropy: a spin state is $t$-anticoherent when all spin moments up to order $t$ are independent of direction, or equivalently when all multipole moments up to rank $t$ vanish. For mixed states, we recast this condition using the symmetric $N=2j$ qubit embedding, where it is equivalently expressed by maximal mixedness of the reduced symmetric $t$-qubit state $\rho_t$. This reformulation led to an axiomatic separation between \emph{total} $t$-anticoherence, which quantifies isotropy independently of its origin, and \emph{quantum} $t$-anticoherence, which isolates the genuinely quantum contribution as a resource monotone associated with a chosen total measure. Their difference defines a \emph{classical} contribution, capturing isotropy produced by statistical mixing.

On the constructive side, we introduced broad families of total measures based on normalized distances from $\rho_t$ to the MMS state, including a Hilbert--Schmidt distance-based measure and a fidelity-based total candidate, as well as simple purity-based and cumulative-multipole measures. The total monotonicity of the fidelity-based candidate under all $\mathrm{SU}(2)$-covariant channels remains supported by numerical evidence. We also showed how cumulative multipole functionals fit naturally into the same axiomatic picture, thereby clarifying the relation with earlier tensor-based descriptions of angular structure. For the quantum contribution, we constructed convex-roof extensions from pure-state functionals tied to bipartite entanglement across the $t\,|\,N{-}t$ split, obtaining explicit purity/linear-entropy-based quantum measures and a fidelity/convex-roof-negativity construction satisfying the quantum-side requirements and the ordering relative to the fidelity-based total candidate. For the constructions where the required ordering is established, the inequality $\mathcal A_t^Q\le \mathcal A_t^T$ ensures a well-defined and nonnegative classical part $\mathcal A_t^C=\mathcal A_t^T-\mathcal A_t^Q$.

Our examples illustrate that mixed-state anticoherence exhibits qualitatively new behavior absent in the pure-state regime. In particular, mixed states supported on $t$-AC subspaces can retain \emph{maximal quantum} $t$-anticoherence independently of the mixing weights, showing that perfect isotropy need not be of classical origin. Conversely, when the achievable anticoherence order is increased at fixed purity, the total isotropy typically becomes more classical: higher-order isotropy is accessible only at sufficiently low purity, and its quantum share generally decreases. We also showed contrasting robustness patterns under particle loss, with HOAP, GHZ, and $W$ states displaying markedly different decay profiles of the quantum contribution.

A conceptual message of this work is that \emph{isotropy of quantum origin} can itself be viewed as a resource, complementary to the usual resource-theoretic viewpoint in which $\mathrm{SU}(2)$-\emph{asymmetry} is the resource. In the asymmetry framework, useful states are those that break rotational symmetry and carry directional information. Here, by contrast, the quantum measure $\mathcal A_t^Q$ isolates the part of isotropy that cannot be reproduced by classical randomization over directions. It therefore identifies states that are direction-insensitive for intrinsically quantum reasons, encoded in genuine multipartite correlations in the symmetric sector. This entanglement-based perspective is already built into our axioms, through the requirements of vanishing on separable symmetric states and monotonicity under LOCC, and is made concrete by the explicit measures introduced here. The resulting distinction is operationally relevant whenever one wishes to separate loss of directional information caused by uncontrolled classical noise from isotropy engineered coherently through entanglement.

Several directions follow naturally. It would be interesting to sharpen the operational role of total, quantum, and classical anticoherence in concrete reference-frame alignment and direction-estimation protocols, and to characterize optimal states under constraints such as purity, rank, or particle loss. On the structural side, a broader classification of $t$-AC subspaces and of mixed states saturating the convex-roof constructions would clarify the landscape of extremal isotropic resources. Finally, although our analysis focuses on permutation-symmetric sectors and $\mathrm{SU}(2)$, the same strategy should extend to other symmetry groups and multipartite settings, providing a general route to separate intrinsically quantum isotropy from isotropy generated by classical uncertainty.

\section*{Acknowledgements}
JM and ESE thank Nicolas Cerf for stimulating discussions on this topic. This work is supported by the Quantum Science and Technology - National Science and Technology Major Project (Grant No.~2025ZD0300801) and the National Natural Science Foundation of China (Grant Nos.~92476201). JM and ESE acknowledge the FWO and the F.R.S.-FNRS for their funding as part of the Excellence of Science programme (EOS project 40007526).
\appendix

\section{Proofs}
\label{App.Proofs}

We begin this Appendix on the proofs of our propositions by two useful results. A general observation is that the ordering axiom \eqref{ax:Q6} follows automatically whenever the chosen total measure is concave.

\begin{proposition}
\label{prop:convex_roof_ordering_general}
Let $\A_t^T:\mathcal D(\mathcal H_j)\to[0,1]$ be a total $t$-anticoherence measure whose restriction to pure states is used to define
\begin{equation}
  \A_t^Q(\rho)=\min_{\{p_i,|\psi_i\rangle\}}
  \sum_i p_i\,\A_t^T(|\psi_i\rangle\langle\psi_i|).
\end{equation}
If $\A_t^T$ is concave on $\mathcal D(\mathcal H_j)$, then
\begin{equation}
  \A_t^Q(\rho)\le \A_t^T(\rho)
  \qquad\forall\,\rho\in\mathcal D(\mathcal H_j),
\end{equation}
i.e., axiom \eqref{ax:Q6} holds.
\end{proposition}

\begin{proof}
Let $\rho=\sum_i p_i|\psi_i\rangle\langle\psi_i|$ be any pure-state decomposition.
By concavity of $\A_t^T$,
\begin{equation}
  \A_t^T(\rho)\ge \sum_i p_i\,\A_t^T(|\psi_i\rangle\langle\psi_i|).
\end{equation}
Since this inequality holds for every decomposition, it also holds for the one minimizing the convex roof, which gives
\begin{equation}
  \A_t^T(\rho)\ge \A_t^Q(\rho).
\end{equation}
\end{proof}

We will also need a result derived in~\cite{2000Vidal}.

\begin{theorem}[\cite{2000Vidal}]
\label{VidalLOCCConstruction}
Let $f:\mathcal D(\mathcal H_{t/2})\to\mathbb R$ be concave and invariant under unitary operations. Then the pure-state functional
\begin{equation}
    A_t^Q(|\psi\rangle)
    =
    f\!\big(\Tr_{N-t}(|\psi\rangle\langle\psi|)\big)
\end{equation}
is non-increasing on average under LOCC operations across the bipartition $t\,|\,N{-}t$. Its convex-roof extension,
\begin{equation}
    A_t^Q(\rho)
    =
    \min_{\{p_i,|\psi_i\rangle\}}
    \sum_i p_i A_t^Q(|\psi_i\rangle),
\end{equation}
is non-increasing under LOCC operations across that bipartition.
\end{theorem}

\subsection{Proof of Prop. \ref{prop1:purity_total}} 
\label{Subsec.Prop1}
\begin{proof}
First, $\operatorname{Tr}(\rho_t^2)=1/(t+1)$ if and only if $\rho_t$ is maximally mixed,
which gives \eqref{ax:T1}. 

Let us now prove \eqref{ax:T2}. If $\rho$ is a pure coherent state, then every reduced state
$\rho_t$ is pure, so $\operatorname{Tr}(\rho_t^2)=1$ and hence $\A^{T,\mathcal{P}}_t(\rho)=0$.
Conversely, if $\A^{T,\mathcal{P}}_t(\rho)=0$, then $\operatorname{Tr}(\rho_t^2)=1$, so $\rho_t$ is pure.
Since $\rho_t$ is obtained by tracing out part of the symmetric $N$-qubit state $\rho_S$,
the purity of $\rho_t$ implies that $\rho_S$ is product across the bipartition
$t\,|\,N{-}t$. Because $\rho_S$ is supported on the fully symmetric subspace, separability
across this bipartition is equivalent to full separability~\cite{Ichikawa2008ExchangeSymmetry}; hence $\rho_S$ is a convex
mixture of symmetric product states. The fact that $\rho_t$ is pure then forces this
mixture to contain only one product component, so
$\rho_S=(|\mathbf n\rangle\langle \mathbf n|)^{\otimes N}$, i.e. $\rho$ is a pure
spin-coherent state. This proves \eqref{ax:T2}.

The $\mathrm{SU}(2)$ invariance \eqref{ax:T3} is inherited from the invariance of
$\Tr(\rho_t^2)$ under global rotations, and normalization to $[0,1]$ is ensured by the
prefactor $(t+1)/t$.

We now prove monotonicity under $\mathrm{SU}(2)$-covariant noise
\eqref{ax:T4}. Let $\Phi$
be an $\mathrm{SU}(2)$-covariant channel on $\mathcal H_j$. By covariance, the
multipole coefficients transform as $\rho_{LM}\mapsto f_L\rho_{LM}$ with
$|f_L|\le1$. For compactness, write
\begin{equation}\label{CNtL}
c_{N,t,L}
=
\frac{t!}{N!}
\sqrt{
\frac{(N-L)!(N+L+1)!}{(t-L)!(t+L+1)!}
}\, ,
\end{equation}
so that Eq.~\eqref{rhotTLM} reads
\begin{equation}
\rho_t=
\frac{\mathbb 1_{t+1}}{t+1}
+
\sum_{L=1}^{t}\sum_{M=-L}^{L}
c_{N,t,L}\rho_{LM}T_{LM}^{(t)} .
\end{equation}
Using the Hilbert--Schmidt normalization of the tensors
$T_{LM}^{(t)}$, we obtain
\begin{equation}
\begin{aligned}
\Tr ([(\Phi(\rho))_t]^2)
&=
\frac{1}{t+1}
+
\sum_{L=1}^t c_{N,t,L}^2 f_L^2
\sum_{M=-L}^L |\rho_{LM}|^2
\\
&\le
\frac{1}{t+1}
+
\sum_{L=1}^t c_{N,t,L}^2
\sum_{M=-L}^L |\rho_{LM}|^2
\\
&=
\Tr(\rho_t^2).
\end{aligned}
\end{equation}
Hence $\A^{T,\mathcal{P}}_t(\Phi(\rho))\ge \A^{T,\mathcal{P}}_t(\rho)$, proving
\eqref{ax:T4}.
\end{proof}

\subsection{Proof of Prop. \ref{prop2:HS_total}} 
\label{Subsec.PropHS_total}

\begin{proof}
Axioms \eqref{ax:T1}--\eqref{ax:T3} follow from the general
distance-based construction. We now prove \eqref{ax:T4} as in Prop.~\ref{prop1:purity_total}. Using the notation \eqref{CNtL} together with the expansion of $\rho_t$ in Eq.~\eqref{rhotTLM}, we have
\begin{equation}
    \begin{aligned}
        \|(\Phi(\rho))_t-\rho_0^{(t)}\|_2^2
        &=
        \sum_{L=1}^{t}\sum_{M=-L}^{L}
        c_{N,t,L}^2 |f_L|^2\,|\rho_{LM}|^2
        \\
        &\le
        \sum_{L=1}^{t}\sum_{M=-L}^{L}
        c_{N,t,L}^2|\rho_{LM}|^2
        \\
        &=
        \|\rho_t-\rho_0^{(t)}\|_2^2,
    \end{aligned}
\end{equation}
because $|f_L|\le 1$ for all $L$. Hence the Hilbert--Schmidt measure
\begin{equation}
\mathcal{A}_t^{T,\mathrm{HS}}(\rho)=1-\frac{\|\rho_t-\rho_0^{(t)}\|_2}{K_t^{\mathrm{HS}}}
\label{eq:HS-app}
\end{equation}
also satisfies axiom \eqref{ax:T4}.
\end{proof}

\subsection{Proof of Prop. \ref{prop:FidelityNegativityRelation}}
\label{Subsec.Prop3}
\begin{proof}
From the Schmidt decomposition, it can be seen that the reduced state of $\ket{\psi}$ on $A$ is given by
\begin{equation}
\rho_A=\sum_{i=1}^{d_A}\alpha_i^2\ket{\phi_A^{(i)}}\bra{\phi_A^{(i)}},
\end{equation}
So, the eigenvalues $\lambda_i$ of $\rho_A$ are related to the Schmidt coefficients $\alpha_i$, that is, $\lambda_i=\alpha_i^2$.
For pure states $\rho=|\psi\rangle\langle\psi|$, the fidelity between $\rho_A$ and the MMS verifies
\begin{equation}
    F \left(\rho_A,\rho_0\right) = \frac{1}{d_A}\left( \Tr \sqrt{\rho_A} \right)^2 =  \frac{1}{d_A}
\left(\sum_{i=1}^{d_A}\sqrt{\lambda_{i}}\right)^{2}.
\label{eq:SumSquaredSchmidtCoefficients_FidelityReduced}
\end{equation}
The eigenvalues of the reduced state $\rho_A$ verify
\begin{equation}
    \left(\sum_{i=1}^{d_A}\sqrt{\lambda_{i}}\right)^{2}=\underbrace{\sum_{i=1}^{d_A}\lambda_{i}}_{=1}+2\sum_{i>j=1}^{d_A}\sqrt{\lambda_{i}\lambda_{j}}
\end{equation}
and from \eqref{eq:Negativity_SchmidtDecomposition}, we see that~\cite{2002Vidal}
\begin{equation}
    \left(\sum_{i=1}^{d_A}\sqrt{\lambda_{i}}\right)^{2} = 1+ 2\Nega_A (\ket{\psi}).
    \label{eq:SumSquaredSchmidtCoefficients_Negativity}
\end{equation}
Injecting \eqref{eq:SumSquaredSchmidtCoefficients_Negativity} in \eqref{eq:SumSquaredSchmidtCoefficients_FidelityReduced}, we obtain our result \eqref{eq:FidelityReduced_Negativity}.
\end{proof}

\subsection{Proof of Prop. \ref{prop:purity_quantum_Q1Q6}}
\label{Subsec.Prop4}

\begin{proof}
Pure-state consistency \eqref{ax:Q2} and $\mathrm{SU}(2)$ invariance \eqref{ax:Q3}
are immediate from the definition and the rotational covariance of the reduced state.
Convexity \eqref{ax:Q4} follows directly from the convex-roof construction.
Axiom \eqref{ax:Q5} follows from Theorem \ref{VidalLOCCConstruction} since the function $f(\rho)=1-\Tr(\rho^2)$ is concave and invariant under unitaries due to the properties of the purity. Axiom \eqref{ax:Q1} follows because fully separable symmetric
states admit decompositions into coherent states, and $\A_t^{\mathcal{P}}$ vanishes on pure
coherent states.

To prove \eqref{ax:Q6}, it is enough to note that the total purity-based measure
$\A_t^{T,\mathcal{P}}(\rho)=\frac{t+1}{t}\big(1-\Tr(\rho_t^2)\big)$ is concave (equivalently, $-\Tr(\rho_t^2)$
is concave because $\Tr(\rho_t^2)$ is convex and $\rho\mapsto\rho_t$ is linear).
Therefore the ordering follows directly from the general
convex-roof ordering Proposition~\ref{prop:convex_roof_ordering_general}.
\end{proof}

\subsection{Proof of Prop. \ref{prop:AtQRRatios}}
\label{Subsec.Prop5}

\begin{proof}
For pure states $|\psi\rangle$, the reduced density matrices across the
bipartitions $t\,|\,N{-}t$ and $(N{-}t)\,|\,t$ have identical nonzero spectra,
which implies $\Tr(\rho_t^2)=\Tr(\rho_{N-t}^2)$ and therefore
\begin{equation}
  \A^{\mathcal{P}}_t(|\psi\rangle)
  = c_t\,\A^{\mathcal{P}}_{N-t}(|\psi\rangle),
  \quad
  c_t=\frac{(t+1)(N-t)}{t(N+1-t)} .
\end{equation}
Let $\rho=\sum_i p_i|\psi_i\rangle\langle\psi_i|$ be an arbitrary pure-state
decomposition.
Using the above relation term by term yields
\begin{equation}
  \sum_i p_i\,\A^{\mathcal{P}}_t(|\psi_i\rangle)
  = c_t\sum_i p_i\,\A^{\mathcal{P}}_{N-t}(|\psi_i\rangle).
\end{equation}
Since this equality holds for every decomposition, it also holds for the
minimizing decompositions defining the convex-roof extensions, giving
\begin{equation}
  \A^{Q,\mathcal{P}}_t(\rho)=c_t\,\A^{Q,\mathcal{P}}_{N-t}(\rho),
\end{equation}
which proves the claim.
\end{proof}

\subsection{Proof of Prop. \ref{prop:distance_quantum_Q1Q6}}
\label{Subsec.Prop6}

\begin{proof}
Axioms~\eqref{ax:Q2} and~\eqref{ax:Q4} follow directly from the convex-roof construction, while axiom~\eqref{ax:Q3} follows from the $\mathrm{SU}(2)$ invariance of the pure-state functional. Axiom~\eqref{ax:Q5} follows by assumption, since $\A_t^d(|\psi\rangle)$ is an entanglement monotone across the bipartition $t\,|\,N{-}t$ and convex roofs of pure-state entanglement monotones are LOCC monotones. Axiom~\eqref{ax:Q1} follows because the pure-state functional vanishes on coherent states, and fully separable symmetric states admit decompositions into coherent states. Finally, if $\A_t^{T,d}$ is concave, axiom~\eqref{ax:Q6} follows from Proposition~\ref{prop:convex_roof_ordering_general}.
\end{proof}

\subsection{Proof of Prop. \ref{prop:Bures_quantum_Q1Q6}}
\label{Subsec.Prop7}

\begin{proof}
We first recall the definition of the total fidelity-based measure:
\begin{equation}
\A^{T,F}_t(\rho)
=
\frac{(\Tr(\sqrt{\rho_t}))^2-1}{t}
=
\frac{(t+1)F(\rho_t,\rho_0^{(t)})-1}{t}.
\end{equation}
As discussed in Sec.~\ref{subsubsec:fidelity_based_total}, this quantity satisfies axioms \eqref{ax:T1}--\eqref{ax:T3}. Indeed, \eqref{ax:T1} follows from
$F(\rho_t,\rho_0^{(t)})=1$ iff $\rho_t=\rho_0^{(t)}$, while \eqref{ax:T2} follows from the fact that $F(\rho_t,\rho_0^{(t)})=1/(t+1)$ iff $\rho_t$ is pure, which in the symmetric sector forces the global state to be a pure spin-coherent state. Axiom \eqref{ax:T3} follows from unitary invariance of fidelity.

We now prove the quantum axioms for
\begin{equation}
\A^{Q,F}_t(\rho)=\frac{2}{t}\Nega_t^{\mathrm{CR}}(\rho).
\end{equation}
Axiom \eqref{ax:Q1} holds because every fully separable symmetric state admits a
decomposition into coherent states, and the negativity vanishes on each coherent
pure state. Axiom \eqref{ax:Q3} follows because global rotations act by local
unitaries with respect to the bipartition $t\,|\,N{-}t$, and therefore do not
change the negativity. Convexity \eqref{ax:Q4} follows directly from the
convex-roof construction, and monotonicity under LOCC operations \eqref{ax:Q5}
follows from the corresponding property of the convex-roof negativity~\cite{Lee2003}.

It remains to verify pure-state consistency \eqref{ax:Q2} and the ordering axiom
\eqref{ax:Q6}. For pure states, Proposition~\ref{prop:FidelityNegativityRelation}
gives
\begin{equation}
F(\rho_t,\rho_0^{(t)})
=
\frac{1+2\Nega_t(|\psi\rangle)}{t+1}.
\end{equation}
Therefore,
\begin{equation}
\begin{aligned}
\A^{T,F}_t(|\psi\rangle) & =
\frac{(t+1)F(\rho_t,\rho_0^{(t)})-1}{t}\\ 
& =
\frac{2}{t}\,\Nega_t(|\psi\rangle)=
\A^{Q,F}_t(|\psi\rangle),
\end{aligned}
\end{equation}
which proves \eqref{ax:Q2}.

For \eqref{ax:Q6}, let
$\rho=\sum_i p_i|\psi_i\rangle\langle\psi_i|$ be an arbitrary pure-state
decomposition and define
$\sigma_i=\Tr_{N-t}(|\psi_i\rangle\langle\psi_i|)$, so that
$\rho_t=\sum_i p_i\sigma_i$. We use the superadditivity of the Schatten
$1/2$-quasinorm~\cite{doi:10.1142/S0129167X1100715X},
\begin{equation}
  \|A+B\|_{1/2}\ge \|A\|_{1/2}+\|B\|_{1/2},
  \qquad
  \|X\|_{1/2}=(\Tr(\sqrt{X}))^2 ,
\end{equation}
for positive semidefinite operators. By positive homogeneity,
\begin{equation}
\begin{aligned}
  (\Tr(\sqrt{\rho_t}))^2
  =
  \left\|\sum_i p_i\sigma_i\right\|_{1/2}
  &\ge
  \sum_i \|p_i\sigma_i\|_{1/2}  \\
  &=
  \sum_i p_i(\Tr(\sqrt{\sigma_i}))^2 .
\end{aligned}
\end{equation}
For each pure state $|\psi_i\rangle$,
\begin{equation}
  (\Tr(\sqrt{\sigma_i}))^2
  =
  1+2\Nega_t(|\psi_i\rangle).
\end{equation}
Hence
\begin{equation}
  (\Tr(\sqrt{\rho_t}))^2
  \ge
  1+2\sum_i p_i\Nega_t(|\psi_i\rangle).
\end{equation}
Since this is true for every pure-state decomposition, it is true in particular
for an optimal decomposition of the convex roof. Thus
\begin{equation}
  (\Tr(\sqrt{\rho_t}))^2
  \ge
  1+2\Nega_t^{\mathrm{CR}}(\rho),
\end{equation}
and consequently
\begin{equation}
  \A^{T,F}_t(\rho)
  =
  \frac{(\Tr(\sqrt{\rho_t}))^2-1}{t}
  \ge
  \frac{2}{t}\Nega_t^{\mathrm{CR}}(\rho)
  =
  \A^{Q,F}_t(\rho).
\end{equation}
This proves \eqref{ax:Q6}.
\end{proof}

\subsection{Proof of Prop.~\ref{prop:cm_total}}
\label{Subsec.Prop8}

\begin{proof}
We first show that $\mathcal A_t^{T,\mathrm{cm}}(\rho)$ is indeed a function taking values in $[0,1]$, which is equivalent to proving that $0 \leq C_{\le t}(\rho) \leq  C_{\le t}(\rho_{\mathrm{coh}})$. Since it is a sum of positive quantities, it is immediate to verify that $C_{\le t}(\rho) \geq 0$. The upper bound is proved in the following lemma:
\begin{lemma}
\label{lem:cm_coherent_bound}
For every spin-$j$ density matrix $\rho$ and every $t=1,\dots,2j-1$,
\begin{equation}
  C_{\le t}(\rho)\le C_{\le t}(\rho_{\rm coh}),
\end{equation}
where $\rho_{\rm coh}$ is any spin-coherent state.
\end{lemma}

\begin{proof}
For pure states, the bound follows from the extremality of SU(2) coherent states for the cumulative multipole distribution: coherent states maximize $C_{\le t}$ at every order $t$~\cite{BjorkKlimovDeLaHozLeuchsSanchezSoto2015}.
For mixed states, write $\rho=\sum_i p_i|\psi_i\rangle\langle\psi_i|$. Since $C_{\le t}$ is the squared Hilbert--Schmidt norm of the projection of $\rho$ onto the subspace spanned by the tensor operators $T_{LM}$ with $1\le L\le t$, it is convex. Hence
\begin{equation}
C_{\le t}(\rho)
\le
\sum_i p_i C_{\le t}(|\psi_i\rangle\langle\psi_i|)
\le
C_{\le t}(\rho_{\rm coh}).
\end{equation}
\end{proof}

We now verify the defining properties of $\mathcal A_t^{T,\mathrm{cm}}(\rho)$. Axiom~\eqref{ax:T1} follows directly from the affine normalization in Eq.~\eqref{totalACcm}: by definition, $C_{\le t}(\rho)=0$ if and only if all multipoles up to rank $t$ vanish, i.e.\ if and only if $\rho$ is $t$-AC. For axiom~\eqref{ax:T2}, $\mathcal A_t^{T,\mathrm{cm}}(\rho)=0$ if and only if $\rho$ is a pure coherent state, using the coherent-state extremality in Lemma~\ref{lem:cm_coherent_bound} together with the uniqueness of the pure-state maximizers. Axiom~\eqref{ax:T3} holds because under a global rotation the coefficients $\rho_{LM}$ transform unitarily within each fixed $L$, so each $a_L(\rho)\equiv\sum_{M=-L}^L |\rho_{LM}|^2$ is invariant, hence so is $C_{\le t}(\rho)$.

It remains to verify axiom~\eqref{ax:T4}. Let $\Phi$ be an $\mathrm{SU}(2)$-covariant channel on $\mathcal H_j$. By the standard structure theorem for covariant channels on a fixed spin irrep (see Sec.~\ref{Sec:TAC} around Eq.~\eqref{fL}), $\Phi$ acts diagonally on irreducible tensor operators:
\begin{equation}
  \Phi(T_{LM})=f_L\,T_{LM},
\end{equation}
with $f_L$ independent of $M$ and satisfying $|f_L|\le 1$.
Therefore, if
\begin{equation}
\rho=\frac{\mathbb 1}{2j+1}+\sum_{L=1}^{2j}\sum_{M=-L}^{L}\rho_{LM}T_{LM},
\end{equation}
then for $\rho'=\Phi(\rho)$ one has $\rho'_{LM}=f_L\rho_{LM}$, and thus
\begin{equation}
  \sum_{M=-L}^{L}|\rho'_{LM}|^2
  =|f_L|^2\sum_{M=-L}^{L}|\rho_{LM}|^2
  \le \sum_{M=-L}^{L}|\rho_{LM}|^2.
\end{equation}
Summing over $L=1,\dots,t$ gives
\begin{equation}
  C_{\le t}(\Phi(\rho))\le C_{\le t}(\rho).
\end{equation}
Since $\mathcal A_t^{T,\mathrm{cm}}$ is an affine decreasing function of $C_{\le t}$, it follows that
\begin{equation}
\mathcal A_t^{T,\mathrm{cm}}(\Phi(\rho))\ge \mathcal A_t^{T,\mathrm{cm}}(\rho),
\end{equation}
which is axiom~\eqref{ax:T4}.
\end{proof}

\subsection{Proof of Prop.~\ref{prop:cm_t1_equiv_purity}}
\label{Subsec.Prop10}

\begin{proof}
For $t=1$, $C_{\le1}(\rho)$ is affine in the one-body reduced purity $r_1=\Tr(\rho_1^2)$ (equivalently, in $|\langle\mathbf J\rangle|^2$) by the standard purity--multipole relations for symmetric states.
Therefore, after normalization to $[0,1]$ with coherent states mapped to $0$ and
$1$-AC pure states to $1$, the pure-state functionals
$\mathcal A^{T,\mathrm{cm}}_1(|\psi\rangle)$ and $\A_1^{\mathcal P}(|\psi\rangle)$ coincide.

Taking convex roofs of the same pure-state functional gives
\begin{equation}
\mathcal A^{Q,\mathrm{cm}}_1(\rho)=\A^{Q,\mathcal{P}}_1(\rho).
\end{equation}
Since $\A^{Q,\mathcal{P}}_1$ satisfies axioms \eqref{ax:Q1}--\eqref{ax:Q6} by
Proposition~\ref{prop:purity_quantum_Q1Q6}, the same is true for
$\mathcal A^{Q,\mathrm{cm}}_1$, in particular \eqref{ax:Q5}.
\end{proof}


\bibliographystyle{apsrev4-2}
\bibliography{refs}

@article{Zhu2025,
  author = {Zhu, Xuanran and Zhang, Chao and An, Zheng and Zeng, Bei},
  title = {Unified framework for calculating convex roof resource measures},
  journal = {npj Quantum Information},
  volume = {11},
  number = {1},
  pages = {56},
  year = {2025},
  month = {April},
  day = {4},
  doi = {10.1038/s41534-025-01012-1},
  url = {https://doi.org/10.1038/s41534-025-01012-1},
  issn = {2056-6387}
}

@article{Lee2003,
  title = {Convex-roof extended negativity as an entanglement measure for bipartite quantum systems},
  author = {Lee, Soojoon and Chi, Dong Pyo and Oh, Sung Dahm and Kim, Jaewan},
  journal = {Phys. Rev. A},
  volume = {68},
  issue = {6},
  pages = {062304},
  numpages = {5},
  year = {2003},
  month = {Dec},
  publisher = {American Physical Society},
  doi = {10.1103/PhysRevA.68.062304},
  url = {https://link.aps.org/doi/10.1103/PhysRevA.68.062304}
}

@article{Nielsen1999,
  title = {Conditions for a Class of Entanglement Transformations},
  author = {Nielsen, M. A.},
  journal = {Phys. Rev. Lett.},
  volume = {83},
  issue = {2},
  pages = {436--439},
  numpages = {0},
  year = {1999},
  month = {Jul},
  publisher = {American Physical Society},
  doi = {10.1103/PhysRevLett.83.436},
  url = {https://link.aps.org/doi/10.1103/PhysRevLett.83.436}
}

@article{Baguette2014,
  title = {Multiqubit symmetric states with maximally mixed one-qubit reductions},
  volume = {90},
  issue = {3},
  pages = {032314},
  numpages = {10},
  year = {2014},
  ISSN = {1094-1622},
  DOI = {10.1103/physreva.90.032314},
  journal = {Phys. Rev. A},
  publisher = {American Physical Society (APS)},
  author = {Baguette,  D. and Bastin,  T. and Martin,  J.},
}

@article{2017Chryssomalakos,
  title = {Optimal quantum rotosensors},
  author = {Chryssomalakos, C. and Hern\'andez-Coronado, H.},
  journal = {Phys. Rev. A},
  volume = {95},
  issue = {5},
  pages = {052125},
  numpages = {5},
  year = {2017},
  month = {May},
  publisher = {American Physical Society},
  doi = {10.1103/PhysRevA.95.052125},
  url = {https://link.aps.org/doi/10.1103/PhysRevA.95.052125}
}

@misc{Holevo2002,
  doi = {10.48550/ARXIV.QUANT-PH/0212025},
  url = {https://arxiv.org/abs/quant-ph/0212025},
  author        = {Holevo, A. S.},
  title         = {Remarks on the classical capacity of quantum channel},
  eprint        = {quant-ph/0212025},
  archivePrefix = {arXiv},
  year          = {2002}
}

@article{2025Goldberg,
  title = {Beyond the Quantum Cram\'er-Rao Bound},
  author = {Hervas, J. R. and Goldberg, A. Z. and Sanz, A. S. and Hradil, Z. and \ifmmode \check{R}\else \v{R}\fi{}eh\'a\ifmmode \check{c}\else \v{c}\fi{}ek, J. and S\'anchez-Soto, L. L.},
  journal = {Phys. Rev. Lett.},
  volume = {134},
  issue = {1},
  pages = {010804},
  numpages = {6},
  year = {2025},
  month = {Jan},
  publisher = {American Physical Society},
  doi = {10.1103/PhysRevLett.134.010804},
  url = {https://link.aps.org/doi/10.1103/PhysRevLett.134.010804}
}

@article{Ichikawa2008ExchangeSymmetry,
  title   = {Exchange symmetry and multipartite entanglement},
  author  = {Ichikawa, Tsubasa and Sasaki, Toshihiko and Tsutsui, Izumi and Yonezawa, Nobuhiro},
  journal = {Physical Review A},
  volume  = {78},
  pages   = {052105},
  year    = {2008},
  doi     = {10.1103/PhysRevA.78.052105}
}

@article{Rivas2013,
  title   = {{SU(2)}-invariant depolarization of quantum states of light},
  author  = {Rivas, {\'A}ngel and Luis, Alfredo},
  journal = {Phys. Rev. A},
  volume  = {88},
  number  = {5},
  pages   = {052120},
  year    = {2013},
  month   = nov,
  doi     = {10.1103/PhysRevA.88.052120},
  url     = {https://doi.org/10.1103/PhysRevA.88.052120}
}

@Inbook{Muller2009,
author="M{\"u}ller, C.A.",
editor="Buchleitner, Andreas
and Viviescas, Carlos
and Tiersch, Markus",
title="Diffusive Spin Transport",
bookTitle="Entanglement and Decoherence: Foundations and Modern Trends",
year="2009",
publisher="Springer Berlin Heidelberg",
address="Berlin, Heidelberg",
pages="277--314",
isbn="978-3-540-88169-8",
doi="10.1007/978-3-540-88169-8_6",
url="https://doi.org/10.1007/978-3-540-88169-8_6"
}

@article{1932Majorana,
  author    = {Ettore Majorana},
  title     = {Atomi orientati in campo magnetico variabile},
  journal   = {Il Nuovo Cimento (1924-1942)},
  year      = {1932},
  volume    = {9},
  number    = {2},
  pages     = {43--50},
  month     = {feb},
  doi       = {10.1007/BF02960953},
  issn      = {1827-6121},
  url       = {https://doi.org/10.1007/BF02960953},
}

@article{2020ESE,
  title = {Majorana representation for mixed states},
  author = {Serrano-Ens\'astiga, E. and Braun, D.},
  journal = {Phys. Rev. A},
  volume = {101},
  issue = {2},
  pages = {022332},
  numpages = {16},
  year = {2020},
  month = {Feb},
  publisher = {American Physical Society},
  doi = {10.1103/PhysRevA.101.022332},
  url = {https://link.aps.org/doi/10.1103/PhysRevA.101.022332}
}

@article{2024Ferretti,
author = {Hugo Ferretti and Y. Batuhan Yilmaz and Kent Bonsma-Fisher and Aaron Z. Goldberg and Noah Lupu-Gladstein and Arthur O. T. Pang and Lee A. Rozema and Aephraim M. Steinberg},
journal = {Optica Quantum},
number = {2},
pages = {91--102},
publisher = {Optica Publishing Group},
title = {Generating a 4-photon tetrahedron state: toward simultaneous super-sensitivity to non-commuting rotations},
volume = {2},
month = {Apr},
year = {2024},
url = {https://opg.optica.org/opticaq/abstract.cfm?URI=opticaq-2-2-91},
doi = {10.1364/OPTICAQ.510125},
}

@Article{2026Denis,
	title={{Coherent generation and protection of anticoherent spin states}},
	author={Jérôme Denis and Colin Read and John Martin},
	journal={SciPost Phys. Core},
	volume={9},
	pages={001},
	year={2026},
	publisher={SciPost},
	doi={10.21468/SciPostPhysCore.9.1.001},
	url={https://scipost.org/10.21468/SciPostPhysCore.9.1.001},
}

@article{2017Bouchard,
author = {F. Bouchard and P. de la Hoz and G. Bj\"{o}rk and R. W. Boyd and M. Grassl and Z. Hradil and E. Karimi and A. B. Klimov and G. Leuchs and J. \v{R}eh\'{a}\v{c}ek and L. L. S\'{a}nchez-Soto},
journal = {Optica},
number = {11},
pages = {1429--1432},
publisher = {Optica Publishing Group},
title = {Quantum metrology at the limit with extremal Majorana constellations},
volume = {4},
month = {Nov},
year = {2017},
url = {https://opg.optica.org/optica/abstract.cfm?URI=optica-4-11-1429},
doi = {10.1364/OPTICA.4.001429},
}

@misc{2026Bringewatt,
      title={Butterfly Echo Protocol for Axis-Agnostic Heisenberg-Limited Metrology}, 
      author={Jacob Bringewatt and Leon Zaporski and Matthew Radzihovsky and Jasmine Albert and Alexey V. Gorshkov and Vladan Vuletic and Gregory Bentsen},
      year={2026},
      eprint={2602.23332},
      archivePrefix={arXiv},
      primaryClass={quant-ph},
      url={https://arxiv.org/abs/2602.23332}, 
}

@book{Var.Mos.Khe:88,
author = {Varshalovich, D A and Moskalev, A N and Khersonskii, V K},
title = {Quantum Theory of Angular Momentum},
publisher = {World Scientific},
year = {1988},
doi = {10.1142/0270},
address = {},
edition   = {},
URL = {https://www.worldscientific.com/doi/abs/10.1142/0270}
}

@article{Goldberg2022,
  title = {From polarization multipoles to higher-order coherences},
  volume = {47},
  ISSN = {1539-4794},
  url = {http://dx.doi.org/10.1364/OL.443053},
  DOI = {10.1364/ol.443053},
  number = {3},
  journal = {Optics Letters},
  publisher = {Optica Publishing Group},
  author = {Goldberg,  Aaron Z. and Klimov,  Andrei B. and deGuise,  Hubert and Leuchs,  Gerd and Agarwal,  Girish S. and Sánchez-Soto,  Luis L.},
  year = {2022},
  month = jan,
  pages = {477}
}

@article{Goldberg2021,
  title = {Quantum concepts in optical polarization},
  volume = {13},
  ISSN = {1943-8206},
  url = {http://dx.doi.org/10.1364/AOP.404175},
  DOI = {10.1364/aop.404175},
  number = {1},
  journal = {Advances in Optics and Photonics},
  publisher = {Optica Publishing Group},
  author = {Goldberg,  Aaron Z. and de la Hoz,  Pablo and Bj\"{o}rk,  Gunnar and Klimov,  Andrei B. and Grassl,  Markus and Leuchs,  Gerd and Sánchez-Soto,  Luis L.},
  year = {2021},
  month = mar,
  pages = {1}
}

@article{Zim:06,
  author	={Zimba, J.},
  title		={Anticoherent Spin States via the {M}ajorana Representation},
  journal	={Electronic Journal of Theoretical Physics},
  volume	={3},
  number	={10},
  pages		={143--156},
  year		={2006}
}

@article{Goldberg2024,
  title = {Robust quantum metrology with random Majorana constellations},
  volume = {10},
  ISSN = {2058-9565},
  url = {http://dx.doi.org/10.1088/2058-9565/ad9ac7},
  DOI = {10.1088/2058-9565/ad9ac7},
  number = {1},
  journal = {Quantum Science and Technology},
  publisher = {IOP Publishing},
  author = {Goldberg,  Aaron Z and Hervas,  Jose R and Sanz,  Angel S and Klimov,  Andrei B and Řeháček,  Jaroslav and Hradil,  Zdeněk and Hiekkam\"{a}ki,  Markus and Eriksson,  Matias and Fickler,  Robert and Leuchs,  Gerd and Sánchez-Soto,  Luis L},
  year = {2024},
  month = dec,
  pages = {015053}
}

@article{Gross_2021,
  title = {Designing Codes around Interactions: The Case of a Spin},
  author = {Gross, Jonathan A.},
  journal = {Phys. Rev. Lett.},
  volume = {127},
  issue = {1},
  pages = {010504},
  numpages = {6},
  year = {2021},
  month = {Jul},
  publisher = {American Physical Society},
  doi = {10.1103/PhysRevLett.127.010504},
  url = {https://link.aps.org/doi/10.1103/PhysRevLett.127.010504}
}

@article{lim_2023,
  title = {Fault-tolerant qubit encoding using a spin-7/2 qudit},
  author = {Lim, Sumin and Liu, Junjie and Ardavan, Arzhang},
  journal = {Phys. Rev. A},
  volume = {108},
  issue = {6},
  pages = {062403},
  numpages = {10},
  year = {2023},
  month = {Dec},
  publisher = {American Physical Society},
  doi = {10.1103/PhysRevA.108.062403},
  url = {https://link.aps.org/doi/10.1103/PhysRevA.108.062403}
}

@article{Bartlett2007,
  author  = {Bartlett, S. D. and Rudolph, T. and Spekkens, R. W.},
  title   = {Reference frames, superselection rules, and quantum information},
  journal = {Rev. Mod. Phys.},
  volume  = {79},
  pages   = {555--609},
  year    = {2007},
  doi     = {10.1103/RevModPhys.79.555}
}

@article{Gour2008,
  author  = {Gour, G. and Spekkens, R. W.},
  title   = {The resource theory of quantum reference frames},
  journal = {New J. Phys.},
  volume  = {10},
  pages   = {033023},
  year    = {2008},
  doi     = {10.1088/1367-2630/10/3/033023}
}

@article{Gour2009,
  author  = {Gour, G. and Marvian, I. and Spekkens, R. W.},
  title   = {Measuring the quality of a quantum reference frame},
  journal = {Phys. Rev. A},
  volume  = {80},
  pages   = {012307},
  year    = {2009},
  doi     = {10.1103/PhysRevA.80.012307}
}

@article{Marvian2013,
  author  = {Marvian, I. and Spekkens, R. W.},
  title   = {Extending Noether's theorem by quantifying the asymmetry of quantum states},
  journal = {Nat. Commun.},
  volume  = {5},
  pages   = {3821},
  year    = {2014},
  doi     = {10.1038/ncomms4821}
}

@article{Marvian2014,
  author  = {Marvian, I. and Spekkens, R. W.},
  title   = {Modes of asymmetry: The application of harmonic analysis to symmetric quantum dynamics},
  journal = {Phys. Rev. A},
  volume  = {90},
  pages   = {014102},
  year    = {2014},
  doi     = {10.1103/PhysRevA.90.014102}
}

@article{CrannPereiraKribs2010,
  author  = {Crann, Jason and Pereira, Rui and Kribs, David W.},
  title   = {Spherical designs and anticoherent spin states},
  journal = {J. Phys. A: Math. Theor.},
  volume  = {43},
  pages   = {255307},
  year    = {2010},
  doi     = {10.1088/1751-8113/43/25/255307}
}

@article{DeLaHozKlimovBjorkLeuchsSanchezSoto2013,
  author  = {de la Hoz, Pablo and Klimov, Andrei B. and Bj\"ork, Gunnar and Kim, Y.-H. and M\"uller, C. and Marquardt, Ch. and Leuchs, Gerd and S\'anchez-Soto, Luis L.},
  title   = {Multipolar hierarchy of efficient quantum polarization measures},
  journal = {Phys. Rev. A},
  volume  = {88},
  number  = {6},
  pages   = {063803},
  year    = {2013},
  doi     = {10.1103/PhysRevA.88.063803}
}

@article{SanchezSotoKlimovDeLaHozLeuchs2013,
  author  = {S\'anchez-Soto, Luis L. and Klimov, Andrei B. and de la Hoz, Pablo and Leuchs, Gerd},
  title   = {Quantum versus classical polarization states: when multipoles count},
  journal = {J. Phys. B: At. Mol. Opt. Phys.},
  volume  = {46},
  number  = {10},
  pages   = {104011},
  year    = {2013},
  doi     = {10.1088/0953-4075/46/10/104011}
}

@article{BjorkKlimovDeLaHozLeuchsSanchezSoto2015,
  author  = {Bj\"ork, Gunnar and Klimov, Andrei B. and de la Hoz, Pablo and Grassl, Markus and Leuchs, Gerd and S\'anchez-Soto, Luis L.},
  title   = {Extremal quantum states and their Majorana constellations},
  journal = {Phys. Rev. A},
  volume  = {92},
  number  = {3},
  pages   = {031801(R)},
  year    = {2015},
  doi     = {10.1103/PhysRevA.92.031801}
}

@article{Kobus2019,
  volume = {100},
  ISSN = {2469-9934},
  url = {http://dx.doi.org/10.1103/PhysRevA.100.032112},
  DOI = {10.1103/physreva.100.032112},
  number = {3},
  pages = {032112},
  journal = {Phys. Rev. A},
  publisher = {American Physical Society (APS)},
  author = {Kłobus,  Waldemar and Burchardt,  Adam and Kołodziejski,  Adrian and Pandit,  Mahasweta and Vértesi,  Tamás and Życzkowski,  Karol and Laskowski,  Wiesław},
  year = {2019},
  month = sep 
}

@Article{Aul.Mar.Mur:10,
  title = {The maximally entangled symmetric state in terms of the geometric measure},
  author = {Aulbach, M. and Markham, D. and Murao, M.},
  journal = {New J. Phys.},
  volume = {12},
  issue = {7},
  pages = {073025},
  year = {2010},
  month = {Jul},
  publisher = {IOP Publishing},
  doi = {10.1088/1367-2630/12/7/073025}
}

@Article{Aul:12,
  title = {Classification of entanglement in symmetric states},
  author = {Aulbach, M.},
  journal = {Int. J. Quantum Inf.},
  volume = {10},
  issue = {7},
  pages = {1230004},
  year = {2012},
  month = {Oct},
  publisher = {World Scientific Publishing Company},
  doi = {10.1142/S0219749912300045}
}

@Article{Bag.Mar:17,
  title = {Anticoherence measures for pure spin states},
  author = {Baguette, D. and Martin, J.},
  journal = {Phys. Rev. A},
  volume = {96},
  issue = {3},
  pages = {032304},
  numpages = {12},
  year = {2017},
  month = {Sep},
  publisher = {American Physical Society},
  doi = {10.1103/PhysRevA.96.032304}
}

@article{PhysRevA.92.052333,
  title = {Anticoherence of spin states with point-group symmetries},
  author = {Baguette, D. and Damanet, F. and Giraud, O. and Martin, J.},
  journal = {Phys. Rev. A},
  volume = {92},
  issue = {5},
  pages = {052333},
  numpages = {13},
  year = {2015},
  month = {Nov},
  publisher = {American Physical Society},
  doi = {10.1103/PhysRevA.92.052333},
  url = {https://link.aps.org/doi/10.1103/PhysRevA.92.052333}
}

@article{Aguilar_2020,
doi = {10.1088/1751-8121/ab6511},
url = {https://doi.org/10.1088/1751-8121/ab6511},
year = {2020},
month = {jan},
publisher = {IOP Publishing},
volume = {53},
number = {6},
pages = {065301},
author = {Aguilar, P and Chryssomalakos, C and Guzmán-González, E and Hanotel, L and Serrano-Ensástiga, E},
title = {When geometric phases turn topological},
journal = {Journal of Physics A: Mathematical and Theoretical}
}

@article{Den.Mar:22,
  title = {Extreme depolarization for any spin},
  author = {Denis, J\'er\^ome and Martin, John},
  journal = {Phys. Rev. Res.},
  volume = {4},
  issue = {1},
  pages = {013178},
  numpages = {24},
  year = {2022},
  month = {Mar},
  publisher = {American Physical Society},
  doi = {10.1103/PhysRevResearch.4.013178},
  url = {https://link.aps.org/doi/10.1103/PhysRevResearch.4.013178}
}

@article{Gir.Bra.Bag.Bas.Mar:15,
  title = {Tensor Representation of Spin States},
  author = {O.\ Giraud and D.\ Braun and D.\ Baguette and T.\ Bastin and J.\ Martin},
  journal = {Phys.\ Rev.\ Lett.{}},
  volume = {114},
  issue = {8},
  pages = {080401},
  numpages = {5},
  year = {2015},
  month = {Feb},
  publisher = {American Physical Society},
  doi = {10.1103/PhysRevLett.114.080401},
  url = {https://link.aps.org/doi/10.1103/PhysRevLett.114.080401}
}

@article{Gis.Pop:99,
  title = {Spin Flips and Quantum Information for Antiparallel Spins},
  author = {Gisin, N. and Popescu, S.},
  journal = {Phys. Rev. Lett.},
  volume = {83},
  issue = {2},
  pages = {432--435},
  numpages = {0},
  year = {1999},
  month = {Jul},
  publisher = {American Physical Society},
  doi = {10.1103/PhysRevLett.83.432},
  url = {http://link.aps.org/doi/10.1103/PhysRevLett.83.432}
}

@article{Gol.Jam:18,
  title = {Quantum-limited Euler angle measurements using anticoherent states},
  author = {Goldberg, Aaron Z. and James, Daniel F. V.},
  journal = {Phys. Rev. A},
  volume = {98},
  issue = {3},
  pages = {032113},
  numpages = {9},
  year = {2018},
  month = {Sep},
  publisher = {American Physical Society},
  doi = {10.1103/PhysRevA.98.032113},
  url = {https://link.aps.org/doi/10.1103/PhysRevA.98.032113}
}

@article{Kolenderski_Demkowicz-Dobrzanski_2008, title={Optimal state for keeping reference frames aligned and the platonic solids}, 
volume={78}, 
ISSN={2469-9934}, 
DOI={10.1103/PhysRevA.78.052333}, 
number={5}, 
journal={Phys.\ Rev.\ A}, 
publisher={aps}, 
author={Kolenderski, Piotr and Demkowicz-Dobrzanski, Rafal}, 
year={2008}, 
pages={052333}
}

@article{Mar.Wei.Gir:20,
  doi = {10.22331/q-2020-06-22-285},
  url = {https://doi.org/10.22331/q-2020-06-22-285},
  title = {Optimal {D}etection of {R}otations about {U}nknown {A}xes by {C}oherent and {A}nticoherent {S}tates},
  author = {Martin, J.{} and Weigert, S.{} and Giraud, O.{}},
  journal = {{Quantum}},
  issn = {2521-327X},
  publisher = {{Verein zur F{\"{o}}rderung des Open Access Publizierens in den Quantenwissenschaften}},
  volume = {4},
  pages = {285},
  month = jun,
  year = {2020}
}

@article{Per.Pau:17,
    author = {Pereira, Rajesh and Paul-Paddock, Connor},
    title = "{Anticoherent subspaces}",
    journal = {J. Math. Phys.},
    volume = {58},
    number = {6},
    pages = {062107},
    year = {2017},
    month = {06},
    issn = {0022-2488},
    doi = {10.1063/1.4986413}
}

@article{Rob.Ber:94,
	author={Robbins, J.{} M.{} and Berry, M.{} V.{}},
	title={A geometric phase for m=0 spins},
	journal={Journal of Physics A: Mathematical and General},
	volume={27},
	number={12},
	pages={L435},
	url={http://stacks.iop.org/0305-4470/27/i=12/a=007},
	year={1994}
}

@article{Toponomic:22,
author = {Chryssomalakos, C. and Hanotel, L. and Guzm\'{a}n-Gonz\'{a}lez, E. and Serrano-Ens\'{a}stiga, E.},
title = {Toponomic quantum computation},
journal = {Mod. Phys. Lett. A},
volume = {37},
number = {27},
pages = {2250184},
year = {2022},
doi = {10.1142/S021773232250184X}
}

@article{2000Vidal,
author = {Guifré Vidal},
title = {Entanglement monotones},
journal = {Journal of Modern Optics},
volume = {47},
number = {2-3},
pages = {355--376},
year = {2000},
publisher = {Taylor \& Francis},
doi = {10.1080/09500340008244048},
URL = {https://www.tandfonline.com/doi/abs/10.1080/09500340008244048},
}

@article{2002Vidal,
  title = {Computable measure of entanglement},
  author = {Vidal, G. and Werner, R. F.},
  journal = {Phys. Rev. A},
  volume = {65},
  issue = {3},
  pages = {032314},
  numpages = {11},
  year = {2002},
  month = {Feb},
  publisher = {American Physical Society},
  doi = {10.1103/PhysRevA.65.032314},
  url = {https://link.aps.org/doi/10.1103/PhysRevA.65.032314}
}

@article{Rudzinski2024orthonormalbasesof,
  doi = {10.22331/q-2024-01-25-1234},
  url = {https://doi.org/10.22331/q-2024-01-25-1234},
  title = {Orthonormal bases of extreme quantumness},
  author = {Rudzi{\'{n}}ski, Marcin and Burchardt, Adam and {\.{Z}}yczkowski, Karol},
  journal = {{Quantum}},
  issn = {2521-327X},
  publisher = {{Verein zur F{\"{o}}rderung des Open Access Publizierens in den Quantenwissenschaften}},
  volume = {8},
  pages = {1234},
  month = jan,
  year = {2024}
}

@article{doi:10.1142/S0129167X1100715X,
author = {Bourin, JEAN-CHRISTOPHE and Hiai, FUMIO},
title = {NORM AND ANTI-NORM INEQUALITIES FOR POSITIVE SEMI-DEFINITE MATRICES},
journal = {International Journal of Mathematics},
volume = {22},
number = {08},
pages = {1121-1138},
year = {2011},
doi = {10.1142/S0129167X1100715X},

URL = { 
    
        https://doi.org/10.1142/S0129167X1100715X
    
    

}
}

@article{PhysRevA.111.022435,
  title = {Quantum metrology of rotations with mixed spin states},
  author = {Serrano-Ens\'astiga, Eduardo and Chryssomalakos, Chryssomalis and Martin, John},
  journal = {Phys. Rev. A},
  volume = {111},
  issue = {2},
  pages = {022435},
  numpages = {14},
  year = {2025},
  month = {Feb},
  publisher = {American Physical Society},
  doi = {10.1103/PhysRevA.111.022435},
  url = {https://link.aps.org/doi/10.1103/PhysRevA.111.022435}
}

@article{S0129055X22500210,
author = {Chang, Euijung and Kim, Jaeyoung and Kwak, Hyesun and Lee, Hun Hee and Youn, Sang-Gyun},
title = {Irreducibly SU(2)-covariant quantum channels of low rank},
journal = {Reviews in Mathematical Physics},
volume = {34},
number = {07},
pages = {2250021},
year = {2022},
doi = {10.1142/S0129055X22500210},
URL = {https://doi.org/10.1142/S0129055X22500210}
}

@article{Aschieri2024,
  author    = {Aschieri, Tommaso and Ruba, B{\l}a\.{z}ej and Solovej, Jan Philip},
  title     = {{SU(2)-Equivariant Quantum Channels: Semiclassical Analysis}},
  journal   = {Communications in Mathematical Physics},
  year      = {2024},
  volume    = {405},
  number    = {12},
  pages     = {298},
  doi       = {10.1007/s00220-024-05178-1},
  url       = {https://doi.org/10.1007/s00220-024-05178-1},
  issn      = {1432-0916}
}

@article{ChitambarGour2019,
  author  = {Chitambar, Eric and Gour, Gilad},
  title   = {Quantum resource theories},
  journal = {Rev. Mod. Phys.},
  volume  = {91},
  pages   = {025001},
  year    = {2019},
  doi     = {10.1103/RevModPhys.91.025001}
}

@article{PeresScudoDirection2001,
  author  = {Peres, Asher and Scudo, Petra F.},
  title   = {Entangled quantum states as direction indicators},
  journal = {Phys. Rev. Lett.},
  volume  = {86},
  pages   = {4160--4163},
  year    = {2001},
  doi     = {10.1103/PhysRevLett.86.4160}
}

@article{PeresScudoCartesian2001,
  author  = {Peres, Asher and Scudo, Petra F.},
  title   = {Transmission of a Cartesian frame by a quantum system},
  journal = {Phys. Rev. Lett.},
  volume  = {87},
  pages   = {167901},
  year    = {2001},
  doi     = {10.1103/PhysRevLett.87.167901}
}

@article{Peres01072002,
author = {Asher Peres and Petra Scudo},
title = {Covariant quantum measurements may not be optimal},
journal = {Journal of Modern Optics},
volume = {49},
number = {8},
pages = {1235--1243},
year = {2002},
publisher = {Taylor \& Francis},
doi = {10.1080/09500340110118449},
URL = { 
    
        https://doi.org/10.1080/09500340110118449
}
}

@article{CollinsPopescu2004,
  author  = {Collins, Daniel and Popescu, Sandu},
  title   = {Frames of reference and the intrinsic directional information of a particle with spin},
  journal = {arXiv preprint quant-ph/0401096},
  year    = {2004},
  eprint  = {quant-ph/0401096},
  archivePrefix = {arXiv}
}

@article{PhysRevResearch.4.013076,
  title = {Improving sum uncertainty relations with the quantum Fisher information},
  author = {Chiew, Shao-Hen and Gessner, Manuel},
  journal = {Phys. Rev. Res.},
  volume = {4},
  issue = {1},
  pages = {013076},
  numpages = {11},
  year = {2022},
  month = {Jan},
  publisher = {American Physical Society},
  doi = {10.1103/PhysRevResearch.4.013076},
  url = {https://link.aps.org/doi/10.1103/PhysRevResearch.4.013076}
}

@article{PhysRevLett.74.1259,
  title = {Optimal Extraction of Information from Finite Quantum Ensembles},
  author = {Massar, S. and Popescu, S.},
  journal = {Phys. Rev. Lett.},
  volume = {74},
  issue = {8},
  pages = {1259--1263},
  numpages = {0},
  year = {1995},
  month = {Feb},
  publisher = {American Physical Society},
  doi = {10.1103/PhysRevLett.74.1259},
  url = {https://link.aps.org/doi/10.1103/PhysRevLett.74.1259}
}

@article{RevModPhys.91.025001,
  title = {Quantum resource theories},
  author = {Chitambar, Eric and Gour, Gilad},
  journal = {Rev. Mod. Phys.},
  volume = {91},
  issue = {2},
  pages = {025001},
  numpages = {48},
  year = {2019},
  month = {Apr},
  publisher = {American Physical Society},
  doi = {10.1103/RevModPhys.91.025001},
  url = {https://link.aps.org/doi/10.1103/RevModPhys.91.025001}
}

\end{document}